\pdfoutput=1
\RequirePackage{ifpdf}
\ifpdf 
\documentclass[pdftex]{sigma}
\else
\documentclass{sigma}
\fi

\newcommand {\C} {{\mathbb{C}}}

\def\nC{\mathbb{C}}
\def\Sing{\text{Sing}}

\usepackage{amscd}
\usepackage[all]{xy}

\numberwithin{equation}{section}

\newtheorem{Theorem}{Theorem}[section]
\newtheorem{Corollary}[Theorem]{Corollary}
\newtheorem{Lemma}[Theorem]{Lemma}
\newtheorem{Proposition}[Theorem]{Proposition}
{ \theoremstyle{definition}
\newtheorem{Definition}[Theorem]{Definition}

\newtheorem{Example}[Theorem]{Example}
\newtheorem{Remark}[Theorem]{Remark} }

\begin{document}
\allowdisplaybreaks

\newcommand{\arXivNumber}{1808.00743}

\renewcommand{\thefootnote}{}

\renewcommand{\PaperNumber}{047}

\FirstPageHeading

\ShortArticleName{Rational KdV Potentials and Differential Galois Theory}

\ArticleName{Rational KdV Potentials \\ and Differential Galois Theory\footnote{This paper is a~contribution to the Special Issue on Algebraic Methods in Dynamical Systems. The full collection is available at \href{https://www.emis.de/journals/SIGMA/AMDS2018.html}{https://www.emis.de/journals/SIGMA/AMDS2018.html}}}

\Author{Sonia JIM\'ENEZ~$^\dag$, Juan J.~MORALES-RUIZ~$^\ddag$, Raquel S\'ANCHEZ-CAUCE~$^\S$\\ and Mar\'ia-\'Angeles ZURRO~$^\S$}

\AuthorNameForHeading{S.~Jim\'enez, J.J.~Morales-Ruiz, R.~S\'anchez-Cauce and M.-A.~Zurro}

\Address{$^\dag$~Junta de Castilla y Le\'on, Salamanca, Spain}
\EmailD{\href{mailto:sonia.jimver@educa.jcyl.es}{sonia.jimver@educa.jcyl.es}}

\Address{$^\ddag$~Departamento de Matem\'atica Aplicada, E.T.S. Edificaci\'on,\\
\hphantom{$^\ddag$}~Universidad Polit\'ecnica de Madrid, Madrid, Spain}
\EmailD{\href{mailto:juan.morales-ruiz@upm.es}{juan.morales-ruiz@upm.es}}

\Address{$^\S$~Departamento de Matem\'aticas, Universidad Aut\'onoma de Madrid, Madrid, Spain}
\EmailD{\href{mailto:raquel.sanchezcauce@predoc.uam.es}{raquel.sanchezcauce@predoc.uam.es}, \href{mailto:mangeles.zurro@uam.es}{mangeles.zurro@uam.es}}

\ArticleDates{Received September 11, 2018, in final form May 29, 2019; Published online June 25, 2019}

\Abstract{In this work, using differential Galois theory, we study the spectral problem of the one-dimensional Schr\"odinger equation for rational time dependent KdV potentials. In particular, we compute the fundamental matrices of the linear systems associated to the Schr\"odinger equation. Furthermore we prove the invariance of the Galois groups with respect to time, to generic values of the spectral parameter and to Darboux transformations.}

\Keywords{differential Galois theory; KdV hierarchy; Schr\"odinger operator; Darboux transformations; spectral curves; rational solitons}

\Classification{12H05; 35Q51; 37K10}

\renewcommand{\thefootnote}{\arabic{footnote}}
\setcounter{footnote}{0}

\section{Introduction}

In 1977 Airault, McKean and Moser studied in \cite{Airault} some special solutions of the KdV equation,
\begin{gather*}
u_t-6uu_x+u_{xxx}=0,
\end{gather*}
like rational and elliptic ones. Then one year later Adler and Moser studied KdV rational solutions of the KdV hierarchy by means of Darboux--Crum transformations, simplifiying the proof of previous results for these solutions \cite{AM}.

One of the goals of the paper is to study the invariance of the Galois group of the linear system
\begin{gather}
\Phi_x = U \Phi = \begin{pmatrix} 0 & 1 \\ u -E & 0 \end{pmatrix} \Phi,\nonumber \\
\Phi_{t_r} = V_r \Phi = \begin{pmatrix} G_r (u) & F_r (u) \\ -H_r (u) & -G_r (u) \end{pmatrix} \Phi,\label{eq:systemkdv}
\end{gather}
associated to the KdV hierarchy, with respect to the Darboux transformations and respect to the KdV flow (i.e., to the time). In fact as a by-product we have obtained more than that: the Galois group is also invariant with respect to generic values of the spectral parameters (see Section~\ref{galo}).

Thus, in some sense this paper can be considered as a continuation of our previous paper~\cite{JMSZ}, where we studied the invariance of the Galois group of the AKNS systems with respect to the Darboux transformations. But one of the essential differences here is that in general we can not use the Darboux invariance result in~\cite{JMSZ}, because the Darboux transformation here is not a well-defined gauge transformation, i.e., it is not inversible. Thus we must use the classical Darboux tranformation of the Schr\"odinger equation, we call it the Darboux--Crum transform; and then to verify the compatibility of this transform with the complete linear system~\eqref{eq:systemkdv}.

In Section \ref{sect-novikovs DT} we study the action of the Darboux transformations over the recursive relations~\eqref{eq:rec_dif_f} inside the KdV hierarchy. We point out that the results in Section~\ref{sect-novikovs DT} hold not only for rational KdV potentials but also for any \textit{arbitrary} KdV potential.

Also, in Section~\ref{sec espec darb} we study the action over the spectral curve of the Darboux transformations for stationary KdV \textit{arbitrary} potentials.

Brezhnev in three papers \cite{BR1,BR2,BR3} also consider the Galois groups associated to spectral problem for some KdV potentials. More specifically the so-called finite-gap potentials, where the spectral curve is non-singular. Here we study a completely different situation, where the spectral curves are cuspidal curves, corresponding to Adler--Moser rational type solutions.

{In some articles, such as \cite{marshall} and~\cite{wilson}, the authors studied the general Schr\"odinger equation, i.e., the potential $u$ is a differential indeterminate which satisfies KdV$_1$ equation. This is not our situation here, since we consider the family of Adler--Moser rational potentials in $1+1$ dimensions. We would like to point out that in the stationary case the results in~\cite{BEG} for algebraically integrable systems proved that the Galois group is contained in a torus at each generic point in the spectral curve, when the field of coefficients is a formal field.}

However, the general results obtained in Sections \ref{sect-novikovs DT} and \ref{sec espec darb} open the door to study more general families of KdV potentials, such as Rosen--Morse potentials or elliptic KdV potentials.

\section{Basic facts on KdV hierarchy}\label{sect2}

{Consider the derivations $\partial_x, \partial_{t_1}, \partial_{t_2}, \ldots, \partial_{t_m}$ with respect to the variables $x$ and $\boldsymbol{t} = (t_1, \ldots,t_m)$. Let $K_r$ be a differential field with compatible derivations $\partial_x$ and $\partial_{t_r}$, with respect to the variables~$x$ and~$t_r$. Let us assume that its field of constants is the field of complex numbers~$\nC$. Let $E \in \nC$ be a complex parameter and $u \in K_r$ be a fixed element of $K_r$.}

Let us consider the differential recursive relations:
\begin{gather} \label{eq:rec_dif_f}
f_0= 1, \qquad f_{j,x} = -\dfrac{1}{4} f_{j-1, xxx} +uf_{j-1,x} + \dfrac{1}{2} u_x f_{j-1},
\end{gather}
see \cite{GD}, where the authors also provided an algorithm to compute $\partial^{-1}_x (f_{j,x})$. Functions $f_j$ are differential polynomials in $u$, see \cite{GD,Ol}. For the first terms one finds
\begin{gather*}
f_0 = 1, \qquad f_1 = \dfrac{1}{2}u +c_1, \qquad f_2 = -\dfrac{1}{8}u_{xx} +\dfrac{3}{8} u^2 +\dfrac{1}{2} c_1 u +c_2,\nonumber\\
f_3 = \dfrac{1}{32}u_{xxxx} -\dfrac{5}{16}u u_{xx} -\dfrac{5}{32} u_x^2 +\dfrac{5}{16} u^3 + c_1 \left (- \dfrac{1}{8}u_{xx} + \dfrac{3}{8} u^2 \right) + \dfrac{1}{2} c_2 u +c_3,
\end{gather*}
for some integration constants $c_i$.

It is well known that the time dependent KdV hierarchy can be constructed as zero curvature condition of the family of integrable systems (see \cite[Chapter~1, Section~2]{GH}):
\begin{gather}\label{eq:systemkdv0}
 {\mathfrak{s}}_r \quad
 \begin{cases}
\Phi_x = U \Phi = \begin{pmatrix} 0 & 1 \\ u -E & 0 \end{pmatrix} \Phi, \\
\Phi_{t_r} = V_r \Phi = \begin{pmatrix} G_r (u) & F_r (u) \\ -H_r (u) & -G_r (u) \end{pmatrix} \Phi,
\end{cases}
\end{gather}
 where $F_r$, $G_r$ and $H_r \in K_r$ are differential polynomials of the potential $u$ defined by
\begin{gather}
\label{eq:F_r} F_r = \sum_{j=0}^r f_{r-j} E^j, \\
\label{eq:G_r} G_r = -\dfrac{F_{r,x}}{2}, \\
\label{eq:H_r} H_r = (E-u)F_r -G_{r,x} = (E-u)F_r + \dfrac{F_{r,xx}}{2}.
\end{gather}
Observe that the degree in $E$ of the matrices $V_r$ and functions $H_r$ is $ r+1$. We point out that the first equation of~\eqref{eq:systemkdv0} is equivalent to the Schr\"odinger equation
\begin{gather}\label{eq:schr0}
(L-E) \phi =(-\partial_{xx} +u -E) \phi = 0
\end{gather}
with $L=-\partial_{xx} +u$.

Its zero curvature condition
\begin{gather*}
U_{t_r} - V_{r,x} + [U,V_r] =0,
\end{gather*}
yields to the KdV$_r$ equation
\begin{gather}\label{eq:kdvn}
\textrm{KdV}_r\colon \quad u_{t_r} = -\dfrac{1}{2} F_{r, xxx} -2 (E-u)F_{r,x} +u_x F_r.
\end{gather}
Using expressions \eqref{eq:rec_dif_f} and \eqref{eq:F_r}, this equation can be rewritten as
\begin{gather}\label{eq:skdvn2}
\textrm{KdV}_r\colon \quad u_{t_r} = 2 f_{r+1,x}.
\end{gather}
We recall that the equation \eqref{eq:skdvn2} is called the level $r$ equation of the KdV hierarchy. {Varying $r \in \mathbb{N}$ we get the KdV hierarchy.} Whenever we want to specify the dependence on the poten\-tial~$u$, we will write $f_j (u)$, $F_j (u)$, $G_j (u)$ and $H_j (u)$ to emphasize this fact.

\subsection{Adler--Moser rational potentials}

In this section we review the KdV$_r$ rational potentials that Adler and Moser constructed in~\cite{AM}. These are a family of rational potentials $u_n $ for Schr\"odinger operator $-\partial_{xx} +u$ of the form $u_n = -2 (\log \theta_n)_{xx} $, where $\theta_n$ are functions in the variables $x$, $t_r$ defined by the differential recursion
\begin{gather}\label{eq:rec_dif}
 \theta_0 = 1, \qquad \theta_1 = x, \qquad \theta_{n+1,x} \theta_{n-1} - \theta_{n+1} \theta_{n-1, x} = (2n+1)\theta^2_n.
\end{gather}

The solutions of this recursion are polynomials in $x$ with coefficients in the field $F=\mathbb{C}(t_r)$. This is an straighforward consequence of the next result, which is an easy extension of the proof of Lemma~2 in~\cite{AM}.

\begin{Lemma}\label{polinomios} Let be $F=\nC(t_r)$, and $a \in \nC^*$, $b \in \nC$. Let $(F[x], \partial_x )$ be the ring of polynomials with derivation $\partial_x $, whose field of constants is~$ F $.
Let consider the sequence defined recursively by
\begin{gather*}
P_0 = 1, \qquad P_1 = ax + b, \qquad P_{n+1,x} P_{n-1} - P_{n+1} P_{n-1,x} = (2n+1)P^2_n.
\end{gather*}
 Then $P_n \in F[x]$ for all $n$.
\end{Lemma}

Now, applying Lemma \ref{polinomios} for $a= 1$ and $ b=0$, we obtain that functions $\theta_n$ are polynomials of $x$ with coefficients in $\nC(t_r)$ for all $n$. We call these polynomials \textsl{Adler--Moser polynomials}.

The first terms of the recursion are
\[ \begin{matrix} n & & \theta_n \\[5pt]
 0 & & 1 \\
 1 & & x \\
 2 & & x^3 + \tau_2 \\
 3 & & x^6 + 5 \tau_2 x^3 + \tau_3 x - 5 \tau_2^2 \end{matrix} \]
with $\tau_j \in \nC(t_r)$ and $\partial_x \tau_j =0$.

\begin{Definition} The functions
\begin{gather}\label{eq:potencial ur}
u_n := -2 (\log \theta_n)_{xx}
\end{gather}
defined by means of Lemma~\ref{polinomios} are called {\it KdV rational solitons}.
\end{Definition}

Adler and Moser proved in Theorem 2 of \cite{AM} that, for each fixed level $r$ of the KdV hierarchy, there exist expressions for $\tau_j \in \nC (t_r)$, $j=2, \ldots,n$, such that each potential $u_n$ defined by means of the formula \eqref{eq:potencial ur} for $\theta_n$ is a solution of the KdV$_r$ equation \eqref{eq:kdvn}, for constants $c_i=0$, $ i=1, \ldots, r$. Hence, the functions $\tau_2,\ldots, \tau_n$ must be adapted in order to get a solution of the KdV$_r$ equation. When this is the case, i.e., when $u_n$ is a solution of the KdV$_r$ equation, we will denote this adjusted potential as $u_{r,n}$ and the corresponding Adler--Moser polynomial as $\theta_{r,n}$ to stress this fact.

\begin{Definition}\label{def-urn}
The functions
\begin{gather*}
u_{r,n} := -2 (\log \theta_{r,n})_{xx}
\end{gather*}
in $\nC (x,t_r)$ defined by means of Lemma~\ref{polinomios} and with the corresponding adjustment of $\tau_j \in \nC (t_r)$, $j=2, \ldots,n$, are called {\it KdV$_r$ rational solitons}.
\end{Definition}

\begin{Example} As an example of adjusted potentials, we show the first Adler--Moser potentials for $r=1$ with the explicit choice of functions $\tau_2,\ldots, \tau_n$. These potentials are solutions of the KdV$_1$ equation for $c_1=0$: $ u_{t_1} = \frac{3}{2} uu_x - \frac{1}{4} u_{xxx} $. The computations were made using SAGE. We have
\[ \begin{matrix} n & \qquad u_{1,n} & \qquad (\tau_2, \dots, \tau_n) \\[7pt]
 0 &\qquad 0 & \\[3pt]
 1 & \qquad \dfrac{2}{x^2} & \\[10pt]
 2 & \qquad \dfrac{6x\big(x^3 -6t_1\big)}{(x^3 +3t_1)^2} & \qquad (3t_1)\\[10pt]
 3 & \qquad \dfrac{6x \big(2x^9 + 675x^3t_1^2 +1350t_1^3\big)}{\big(x^6 + 15x^3t_1 -45t_1^2\big)^2} & \qquad (3t_1,0) \\[10pt]
 4 & \qquad \dfrac{10 p_4 (x,t_1)}{\big(x^{10} + 45 x^7 t_1 + 4725 x t_1^3 \big)^2} & \qquad (3t_1, 0,0) \\[10pt]
 5 & \qquad \dfrac{30x p_5 (x,t_1)}{(\theta_5)^2} & \qquad \big(3t_1, 0,0, 33075 t_1^3 \big)
 \end{matrix} \]
where
\begin{gather*}
 p_4 (x,t_1) = 2 x^{18} + 72 x^{15} t_1 + 2835 x^{12} t_1^2 - 66150 x^9 t_1^3 - 1190700 x^6 t_1^4 + 4465125 t_1^6,\\
 p_5 (x,t_1) = x^{27} + 126 x^{24} t_1 + 7560 x^{21} t_1^2 + 5655825 x^{15} t_1^4 + 500094000 x^{12} t_1^5 \\
\hphantom{p_5 (x,t_1) =}{} +4313310750 x^9 t_1^6 + 11252115000 x^6 t_1^7 + 295368018750 x^3 t_1^8 -590736037500 t_1^9, \\
 \theta_5 = x^{15} + 105 x^{12} t_1 + 1575 x^9 t_1^2 + 33075 x^6 t_1^3 - 992250 x^3 t_1^4 -1488375 t_1^5.
\end{gather*}
We notice that the adjustment of $\tau_i$ is not linear in $t_1$.
\end{Example}

\subsection{Spectral curves for KdV hierarchy}\label{subsect: spec curve}

Next, we consider the stationary KdV hierarchy. Let $u^{(0)} (x) = u (x, t_r =0)$ be an arbitrary stationary potential. The associated linear system, corresponding to system \eqref{eq:systemkdv0}, will be
\begin{gather}
\Phi_x = U^{(0)} \Phi = \begin{pmatrix} 0 & 1 \\ u^{(0)} -E & 0 \end{pmatrix} \Phi,\nonumber \\
\Phi_{t_r} = V^{(0)}_r \Phi = \begin{pmatrix} G_r \big(u^{(0)}\big) & F_r \big(u^{(0)}\big) \\ -H_r \big(u^{(0)}\big) & -G_r \big(u^{(0)}\big) \end{pmatrix} \Phi.\label{eq:systemkdv0 est 1}
\end{gather}
To simplify the notation, from now on we write $F^{(0)}_r$, $G^{(0)}_r$ and $H^{(0)}_r$ instead of $F_r \big(u^{(0)}\big)$, $G_r \big(u^{(0)}\big)$ and $H_r \big(u^{(0)}\big)$.
The zero curvature condition of this system is now the stationary KdV$_r$ equation
\begin{gather}\label{eq:kdv est}
\textrm{s-KdV}_r\colon \quad 0= -\dfrac{1}{2} F^{(0)}_{r, xxx} -2 \big(E-u^{(0)}\big)F^{(0)}_{r,x} +u^{(0)}_{x} F^{(0)}_r.
\end{gather}
After applying expressions \eqref{eq:rec_dif_f} and \eqref{eq:F_r}, this equation can be rewritten as
\begin{gather*}
\textrm{s-KdV}_r\colon \quad 0 = 2 f_{r+1,x} \big(u^{(0)}\big) = 2 f^{(0)}_{r+1,x}.
\end{gather*}

When the potential $u^{(0)}$ is a solution of the zero curvature condition \eqref{eq:kdv est} we will say that it is a s-KdV$_r$ potential. Under this assumption, the spectral curve of system \eqref{eq:systemkdv0 est 1} for this potential is the characteristic polynomial of matrix $iV^{(0)}_r$:
\begin{align}
\Gamma_r \colon \ \det \big(\mu I_2 - iV^{(0)}_r\big) & = \mu^2 + \big(G^{(0)}_r\big)^2 - F^{(0)}_r H^{(0)}_r \nonumber \\
& = \mu^2 - \dfrac{F^{(0)}_r F^{(0)}_{r,xx}}{2} +\big(u^{(0)}-E\big) \big(F^{(0)}_r\big)^2 + \dfrac{\big(F^{(0)}_{r,x}\big)^2}{4} \nonumber \\
& = \mu^2 -R_{2r+1}(E) =0.\label{eq:spectral curve}
\end{align}
(see for instance \cite{GH} for a general definition of spectral curve). We denote by $p_r (E,\mu ) = \mu^2 - R_{2r+1}(E)$ the equation that defines the spectral curve. We will use the following notation
\begin{gather*}
 R_{2r+1}(E) = \sum_{i=0}^{2r+1} C_i E^i,
\end{gather*}
where $C_i$ are differential polynomials in $u^{(0)}$ with constant coefficients.

\begin{Lemma}\label{lem:coeff c0}
We have the following equality
$\partial_x C_0 = -2f_r f_{r+1,x}. $
\end{Lemma}

\begin{proof}
Replacing $E=0$ in \eqref{eq:spectral curve} we find
\begin{gather*}
R_{2r+1} (0) = C_0 = \frac{- f_{r,x} f_{r,x}}{4} + \frac{f_{r} f_{r,xx}}{2} - u^{(0)} f_{r} f_{r}.
\end{gather*} By derivating with respect to $x$ and using formula \eqref{eq:rec_dif_f} we arrive to the required expression.
\end{proof}

With this matrix presentation it is easy to prove the following result due to Burchnall and Chaundy~\cite{BC1}:

\begin{Proposition}[{\cite[Section~II, p.~560]{BC1}}]
Let $u = u (x)$ be solution of equation~\eqref{eq:kdv est}, we have that $p( E,\mu) = \mu^2 - R_{2r+1}(E) \in \nC[\mu, E]$. Moreover, $R_{2r+1}(E)$ is a polynomial of degree $2r+1$ in~$\nC [E]$.
\end{Proposition}

\begin{Remark}
A potential $u$ can be a solution of several equations of the KdV hierarchy. Therefore, for each level considered, there would be a different spectral curve for the same potential. This ambiguity is clarified when the corresponding Schrodinger operator's centralizer is con\-si\-dered. Furthermore, this centralizer is isomorphic to the ring of rational functions of an algebraic plane curve: the spectral curve that corresponds to the first level of the hierarchy of which the potential~$u$ is a solution. See~\cite{MoRuZu2}.
\end{Remark}

This proposition together with Lemma~\ref{lem:coeff c0} and relation~\eqref{eq:rec_dif_f} yields to the following result.

\begin{Corollary}\label{cor:curva y pot}
Let $\mu^2 - R_{2r+1}(E)=0$ be the spectral curve for potential $u^{(0)}$. If the degree of $R_{2r+1}(E)$ is $2r+1$ in $E$ then, $u^{(0)}$ is solution of a s-KdV$_r$ equation.
\end{Corollary}

Now, we consider the Adler--Moser potentials $u_{r,n}$. We have the following result in the stationary case \cite{AM}:

\begin{Lemma} \label{lema:pot estacionarios}
For $\tau_j =0$, $ j=2, \ldots,n$, the Adler--Moser polynomials and potentials become
\begin{gather*}
\theta^{(0)}_n (x) = \theta_n (x, 0) = x^{n(n+1)/2} \qquad \textrm{and} \qquad u^{(0)}_{r,n} (x) =u_{r,n} (x,{t}_r = 0)= n(n+1)x^{-2}.
\end{gather*}
\end{Lemma}

For a fixed $n$, potential $u^{(0)}_{r,n} (x)= n(n+1)x^{-2}$ defined in the aforementioned lemma is solution of the level~$n$ equation of the stationary KdV hierarchy, the s-KdV$_n$ equation. This implies that in the stationary case we will have $r=n$, i.e., for these s-KdV potentials the iteration level of the recursion~\eqref{eq:rec_dif} is the same as the s-KdV level. For this reason, from now on we will denote the stationary Adler--Moser potentials just by $u^{(0)}_n (x)$ and we will refer to level $n$ stationary KdV equation (instead of level~$r$):
\begin{gather*}
\textrm{s-KdV}_n\colon \quad 0 = 2 f^{(0)}_{n+1,x}.
\end{gather*}

It is well known that the spectral curve associated to system \eqref{eq:systemkdv0 est 1} for these Adler--Moser stationary potentials are
\begin{gather*}
 \Gamma_n \colon \ p_n (E,\mu )= \mu^2 - E^{2n+1} =0.
\end{gather*}
Therefore, we will associate these curves corresponding to the stationary situation, to system~\eqref{eq:systemkdv0} for Adler--Moser potentials~$u_{r,n}$.

\begin{Remark}If we take the potential $u_{r,n}$ solution of KdV$_r$ equation, then the poten\-tial~$u^{(0)}_n (x)$ is a solution of the s-KdV$_n$ equation. Thus, we can link the level~$r$ of the time-dependent KdV hierarchy with the level~$n$ of the stationary KdV hierarchy.
\end{Remark}

\section[Darboux transformations for $f_j$]{Darboux transformations for $\boldsymbol{f_j}$}\label{sect-novikovs DT}

In this section we establish a series of results that will allow us to perform Darboux transformations to KdV differential systems~\eqref{eq:systemkdv0} in the case we have particular solutions at energy level zero. In this way, we can extend the techniques to compute matrix Darboux transformations developed, for instance, in~\cite{GHZ} to the only case where they are not valid: $E=0$.

For that, we will consider the classical Darboux--Crum transformations for the Schr\"odinger equation and we will present the behaviour of these transformations acting on the differential polynomials~$f_j (u)$.

Let us consider the Schr\"odinger equation
\begin{gather}\label{eq:schrE0}
(L-E_0) \phi =(-\partial_{xx} +u -E_0) \phi = 0,
\end{gather}
where $E_0$ is a fixed energy level. Let $\phi_0$ be a solution of such equation. Recall that a {\it Darboux transformation of a function $\phi$ by $\phi_0$} is defined by the formula
\[
{\rm DT}(\phi_0) \phi= \phi_x - \dfrac{\phi_{0,x}}{\phi_0} \phi.
\]
Then the transformed function $\widetilde{\phi}={\rm DT}(\phi_0) \phi$ is a solution of the Schr\"odinger equation for potential $ \widetilde{u}=u - 2 (\log \phi_0)_{xx}$, whenever $\phi$ is a solution of Schr\"odinger equation for potential $u$ and energy level $E \neq E_0$~\cite{CRUM,da1,da2,Salle}. We will denote by ${\rm DT}(\phi_0) u $ the potential $\widetilde{u}$ to point out the fact that it depends on the choice of $\phi_0$.

Next we can observe that the Riccati equation
\begin{gather}\label{eq:riccati sigma}
 \sigma_x = u -E_0 - \sigma^2
\end{gather}
has $\sigma_0 = (\log \phi_0)_x$ as solution, and then
\begin{gather}\label{eq:chapter3 DT}
 {\rm DT}(\phi_0) u = u - 2 \sigma_{0,x}.
\end{gather}
In this way, we retrieve a Riccati equation for $\widetilde{u}$:
\begin{gather*}
 \widetilde{u} = u- 2 \sigma_x = \big(\sigma_x + E_0 + \sigma^2\big) - 2 \sigma_x = \sigma^2 - \sigma_x + E_0.
\end{gather*}
Moreover, whenever we have a solution $\phi$ of the Schr\"odinger equation \eqref{eq:schr0}, the formula $\sigma = (\log \phi)_x$ gives a solution of the Riccati equation \eqref{eq:riccati sigma}. Hence, $\sigma$ satisfies the nonlinear differential equation
\begin{gather}\label{eq:riccati sigma x}
 \sigma_{xx} = u_x -2 \sigma \sigma_x.
\end{gather}

Next, we consider the matrix differential system \eqref{eq:systemkdv0}. Then we perform a Darboux transformation, ${\rm DT}(\phi_0)$, on it obtaing a new differential system, say $\Phi_x = \widetilde{U} \Phi$, $\Phi_{t_r} = \widetilde{V}_r \Phi$, whose zero curvature condition is still equation~\eqref{eq:kdvn}. Let $F_r ( \widetilde{u})$, $G_r( \widetilde{u})$ and $H_r ( \widetilde{u})$ be the corresponding entries of the matrix $ \widetilde{V}_r$. These differential polynomials are given by expressions~\eqref{eq:F_r}, \eqref{eq:G_r} and~\eqref{eq:H_r} in terms on $f_j (\widetilde{u})$. We will establish the relation between $f_j (\widetilde{u})$ and $f_j ({u})$ in the next theorem.

\begin{Theorem}\label{lemma: darboux fi}Let $\phi$ be a solution of Schr\"odinger equation \eqref{eq:schrE0}. Let be $\sigma = (\log \phi)_x $ and $\widetilde{u} = u - 2 \sigma_x$ the Darboux transformed of $u$ by $\phi$. Then, we have
\[
f_j (\widetilde{u}) = f_j (u) + A_j, \qquad \text{for} \quad j=0, 1, 2, \ldots,
\]
where $A_j $ is a differential polynomial in $u$ and $\sigma$. Moreover, $A_j$ satisfies the recursive differential relations
\begin{enumerate}\itemsep=0pt
\item[$1)$] $A_j = -\frac{1}{4} A_{j-1,xx}+ u A_{j-1} - \frac{3}{2} \sigma_x A_{j-1} - \sigma_x f_{j-1} (u)$ and
\item[$2)$] $A_{j,x}+2 \sigma A_j +2 f_{j,x} (u) =0 $.
\end{enumerate}
\end{Theorem}

\begin{proof} We will proceed by induction on $n$.

First, we prove by induction that $f_j (\widetilde{u}) = f_j (u)+ A_j$. For $j=0$ we have $f_0 (\widetilde{u}) = 1 = f_0 (u) + A_0$, where $A_0 =0$. We suppose it is true for $j$ and we prove it for $j+1$. By applying equation \eqref{eq:rec_dif_f} and induction hypothesis we find
\begin{align*}
 f_{j+1,x} (\widetilde{u}) ={}& -\dfrac{1}{4} f_{j,xxx} (\widetilde{u}) + \widetilde{u} f_{j,x} (\widetilde{u}) + \dfrac{1}{2} \widetilde{u}_x f_{j} (\widetilde{u}) \\
 ={}& -\dfrac{1}{4} f_{j,xxx} (u) +u f_{j,x} (u) + \dfrac{1}{2} u_x f_j (u) -\dfrac{1}{4} A_{j,xxx} +u A_{j,x} -2f_{j,x} (u) \sigma_x \\
 & {}-2 A_{j,x} \sigma_x + \dfrac{1}{2} u_x A_j - f_j (u) \sigma_{xx} - A_j \sigma_{xx} = f_{j+1,x} (u) + A_{j+1,x},
\end{align*}
for
\begin{gather}\label{eq:chapter3 A_j}
 A_{j+1,x} = -\dfrac{A_{j,xxx}}{4} +u A_{j,x} -2f_{j,x} (u) \sigma_x -2 A_{j,x} \sigma_x + \dfrac{ u_x A_j}{2} - f_j (u) \sigma_{xx} - A_j \sigma_{xx}.
\end{gather}
Thus, $f_{j+1} (\widetilde{u}) = f_{j+1} (u) + A_{j+1}$ as we wanted to prove.

Now, we prove statements 1 and 2. We do it by induction and simultaneously. Since $A_0 =0$ and $f_0 (u) = f_0 (\widetilde{u}) = 1$, the case $j=0$ is the trivial one. So, we start the induction process in $j=1$. For this, by using recursion formula~\eqref{eq:rec_dif_f} we have
\[ f_{1,x} (\widetilde{u})= -\dfrac{1}{4} f_{0,xxx} (\widetilde{u}) + \widetilde{u} f_{0,x} (\widetilde{u}) + \dfrac{1}{2} \widetilde{u}_x f_{0} (\widetilde{u}) =\dfrac{1}{2} \widetilde{u}_x. \]
Hence, $f_1 (\widetilde{u}) = \frac{\widetilde{u}}{2} + c_1 = \frac{u}{2} - \sigma_x + c_1 = f_1 (u) - \sigma_x$, then $A_1 = - \sigma_x$. For $j=1$ statements~1 and~2 read
\begin{enumerate}\itemsep=0pt
 \item[1)] $ -\frac{1}{4}A_{0,xx} + u A_0 - \frac{3}{2} \sigma_x A_0 - \sigma_x f_0 (u) = - \sigma_x = A_1 $ and
 \item[2)] $ -2 f_{1,x} (u) - A_{1,x} = -u_x + \sigma_{xx} = -2 \sigma \sigma_x = 2 \sigma A_1 $,
\end{enumerate}
by equation \eqref{eq:riccati sigma x}. Now, we suppose the both statements are true for $j$ and we prove them for $j+1$. Derivation with respect to $x$ in the right hand side of statement~1 yields to
\begin{gather*}
 -\dfrac{A_{j,xxx}}{4} + u_x A_{j} + u A_{j,x} - \dfrac{3}{2} \sigma_{xx} A_{j} - \dfrac{3}{2} \sigma_{x} A_{j,x} - \sigma_{xx} f_{j} (u) - \sigma_x f_{j,x} (u) \\
\qquad{} = -\dfrac{A_{j,xxx}}{4} + u A_{j,x} - \sigma_{xx} f_{j} (u) -\sigma_{xx} A_{j} - \dfrac{ \sigma_{xx} A_j}{2} + u_x A_j - \dfrac{3}{2} \sigma_{x} A_{j,x} - \sigma_x f_{j,x} (u).
\end{gather*}
Applying equality \eqref{eq:riccati sigma x} to the term $\sigma_{xx} A_j /2$ we get
\begin{gather*}
 -\dfrac{A_{j,xxx}}{4} + u A_{j,x} - \sigma_{xx} f_{j} (u) -\sigma_{xx} A_{j} - \dfrac{u_x A_j - 2\sigma \sigma_x A_j}{2} + u_x A_j - \dfrac{3}{2} \sigma_{x} A_{j,x} - \sigma_x f_{j,x} (u) \\
 = - \dfrac{A_{j,xxx}}{4} + u A_{j,x} - \sigma_{xx} f_{j} (u) -\sigma_{xx} A_{j} + \sigma \sigma_x A_j + \dfrac{u_x A_j}{2} - \dfrac{3}{2} \sigma_{x} A_{j,x} - \sigma_x f_{j,x} (u) \\
 = - \dfrac{A_{j,xxx}}{4} + u A_{j,x} - \sigma_{xx} f_{j} (u) -\sigma_{xx} A_{j} + \dfrac{u_x A_j}{2} - 2 \sigma_{x} A_{j,x} - \sigma_x f_{j,x} (u) +\sigma_x ( \sigma A_j + \tfrac{1}{2} A_{j,x} ).
\end{gather*}
Applying induction hypothesis for statement 2 we have
\begin{gather*}
 -\dfrac{A_{j,xxx}}{4} + u A_{j,x} - \sigma_{xx} f_{j} (u) -\sigma_{xx} A_{j} + \dfrac{u_x A_j}{2} - 2 \sigma_{x} A_{j,x} - \sigma_x f_{j,x} (u) -\sigma_x f_{j,x} (u) \\
\qquad{} = -\dfrac{A_{j,xxx}}{4} + u A_{j,x} - \sigma_{xx} f_{j} (u) -\sigma_{xx} A_{j} + \dfrac{u_x A_j}{2} - 2 \sigma_{x} A_{j,x} -2 \sigma_x f_{j,x} (u),
\end{gather*}
which is exactly expression \eqref{eq:chapter3 A_j} for $A_{j+1,x}$. So, we can assume that{\samepage
\[ A_{j+1} = -\dfrac{A_{j,xx}}{4} + u A_{j} - \dfrac{3}{2} \sigma_x A_{j} - \sigma_x f_{j} (u). \]
Thus, statement 1 is proved.}

Finally, by equations \eqref{eq:rec_dif_f}, \eqref{eq:chapter3 A_j}, \eqref{eq:riccati sigma x} and induction hypothesis we find for statement~2
\begin{align*}
 -2 f_{j+1,x} - A_{j+1,x} ={}& \dfrac{f_{j,xxx} (u)}{2} -2u f_{j,x} (u) - u_x f_j (u) + \dfrac{A_{j,xxx}}{4} -u A_{j,x} +2f_{j,x} (u) \sigma_x \\
 &{}- \dfrac{ u_x A_j}{2} +2 A_{j,x} \sigma_x + f_j (u) \sigma_{xx} + A_j \sigma_{xx} \\
 ={}& \left ( \dfrac{f_{j,x} (u)}{2} + \dfrac{A_{j,x}}{4} \right )_{xx} + (-2f_{j,x} (u) - A_{j,x}) (u - \sigma_x) -u_x f_j (u) \\
 & {}- \dfrac{u_x A_j}{2} + A_{j,x} \sigma_x + f_j (u) \sigma_{xx} + A_j \sigma_{xx} \\
={} & - \dfrac{\sigma A_{j,xx}}{2} +2u \sigma A_j + A_j \left ( \dfrac{\sigma_{xx}}{2}- \dfrac{u_x }{2} - 2 \sigma \sigma_x \right ) + f_j (u) (\sigma_{xx} -u_x) \\
={} & - \dfrac{\sigma A_{j,xx}}{2} +2u \sigma A_j -3 A_j \sigma \sigma_x -2 f_j (u) \sigma \sigma_x \\
={} & 2 \sigma \left ( -\dfrac{A_{j,xx}}{4} + u A_{j} - \dfrac{3}{2} \sigma_x A_{j} - \sigma_x f_{j} (u) \right) = 2 \sigma A_{j+1}
\end{align*}
by statement 1. Therefore, statement 2 is also proved. This completes the proof.
\end{proof}

\begin{Example}To illustrate the previous theorem we will consider the following KdV$_2$ potentials in the system~\eqref{eq:systemkdv0}.

Let us take
\begin{gather*}
u= \frac{6\big(2x^{10} +270x^5 t_2 + 675t_2^2\big)}{x^2 \big(x^5 - 45t_2\big)^2}
\end{gather*} and the solution $\phi_0= \frac{x^2}{x^5 - 45t_2}$. Then $ \widetilde{u}=\frac{6}{x^2}$. Observe that
\begin{gather*}
 f_1 (u) = \dfrac{u}{2} = \dfrac{3\big(2x^{10} +270x^5 t_2 + 675t_2^2\big)}{x^2 \big(x^5 - 45t_2\big)^2}, \qquad
 f_2 (u)= -\dfrac{u_{xx}}{8} + \dfrac{3}{8}u^2 = \dfrac{45x\big(x^5 +30 t_2\big)}{\big(x^5 - 45t_2\big)^2},
\end{gather*}
and also
\begin{gather*}
 f_1 (\widetilde{u}) = \dfrac{\widetilde{u}}{2} = \dfrac{3}{x^2}, \qquad
 f_2 (\widetilde{u})= -\dfrac{\widetilde{u}_{xx}}{8} + \dfrac{3}{8}\widetilde{u}^2 = \dfrac{9}{x^4}.
\end{gather*}
Hence, in this case
\begin{gather*}
A_1 = f_1 (\widetilde{u}) - f_1 (u) = \dfrac{-3\big(x^{10} + 360x^5 t_2 - 1350t_2^2\big)}{x^2\big(x^5 - 45t_2\big)^2},\\
 A_2 = f_2 (\widetilde{u}) - f_2 (u) = \dfrac{-9\big(4x^{10} +240 x^5t_2 -2025t_2^2\big)}{x^4 \big(x^5 -45t_2\big)^2}.
\end{gather*}
By a direct computation we can verify that the $A_j$ satisfy the relations~1 and~2 of Theorem \ref{lemma: darboux fi}.
\end{Example}

\begin{Corollary}For $i\geq j$ we have the following equality
\begin{gather*}
\sum_{j=0}^i (2 \sigma A_{i-j} + 2f_{i-j,x} (u) + A_{i-j,x}) E^j =0.
\end{gather*}
\end{Corollary}

Theorem~\ref{lemma: darboux fi} has several interesting consequences. The main ones are the relations that the transformed potential $\widetilde{u}$ produce for functions $F_r (u)$. Next we stablish some of them, which will be used in the following sections. In particular, Proposition~\ref{PROP-sigmaT} is specially interesting since it gives a relation between $\sigma_x$ and $\sigma_{t_r }$.

\begin{Proposition}\label{PROP-potT}
\label{cor: darboux fi}
Let $A_i$ and $\sigma$ be as in Theorem~{\rm \ref{lemma: darboux fi}}. For $i= 0, 1, 2, \ldots$ we have
\begin{enumerate}\itemsep=0pt
 \item[$1.$] $ F_i (\widetilde{u}) = F_i (u) + P_i$, where $P_i = \sum\limits_{j=0}^i E^j A_{i-j} $.
 \item[$2.$] Moreover $ P_{i,x} + 2 \sigma P_i + 2F_{i,x} (u) =0 $.
\end{enumerate}
\end{Proposition}

\begin{proof}It is an immediate consequence of Theorem~\ref{lemma: darboux fi}.
\end{proof}

\begin{Proposition}\label{PROP-sigmaT}
Let $u$ be a solution of KdV$_r$ equation. Let $\phi$ be a solution of Schr\"odinger equation~\eqref{eq:schr0} for potential $u$ and energy $E_0$. Let be $\sigma = (\log \phi)_x $. Consider $A_{r+1}$ as defined in Theorem~{\rm \ref{lemma: darboux fi}} and $P_r$ as defined in Proposition~{\rm \ref{cor: darboux fi}}. Then, we have
\begin{gather}\label{eq:chapter3 sigma_tr}
 \sigma_{t_r} = -A_{r+1} = \dfrac{1}{4} P_{r,xx} + E P_{r} + \sigma_{x} F_r (u) + \frac{1}{2} P_r (-2u + 3 \sigma_x).
\end{gather}
\end{Proposition}

\begin{proof}We compare the zero curvature conditions for $u$ and $\widetilde{u}$:
\begin{gather*}
 u_{t_r} = 2 f_{r+1,x} (u)= -\frac{1}{2} F_{r,xxx} (u) +2 (u -E) F_{r,x} (u) + u_x F_r (u), \\
 \widetilde{u}_{t_r} = 2 f_{r+1,x} (\widetilde{u})= -\frac{1}{2} F_{r,xxx} (\widetilde{u}) +2 (\widetilde{u} -E) F_{r,x} (\widetilde{u}) + \widetilde{u}_x F_r (\widetilde{u}).
\end{gather*}
We prove the first equality. For this, we have
\begin{gather*} \widetilde{u}_{t_r} = (u -2 \sigma_x)_{t_r} = u_{t_r} - 2 \sigma_{x,t_r} \qquad \text{and} \qquad
 2 f_{r+1,x} (\widetilde{u}) = 2 f_{r+1,x} (u) + 2 A_{r+1,x}\end{gather*} by Theorem \ref{lemma: darboux fi}. Then
\[ 2 \sigma_{x,t_r} = u_{t_r} - \widetilde{u}_{t_r} = 2 f_{r+1,x} (u) - 2 f_{r+1,x} (\widetilde{u}) = -2 A_{r+1,x}. \]
Thus, $\sigma_{t_r} = -A_{r+1} $.

Now, we prove the second equality.
Using expression \eqref{eq:chapter3 DT} for $\widetilde{u}$ and applying \ref{cor: darboux fi} (1), we obtain
\begin{gather*}
\widetilde{u}_{t_r} = -\dfrac{1}{2} F_{r,xxx} (u) +2 (u -E) F_{r,x} (u) + u_x F_r (u) - \dfrac{1}{2} P_{r,xxx} -2 (E-u) P_{r,x} \\
 \hphantom{\widetilde{u}_{t_r} =}{} -4 \sigma_x F_{r,x} (u) -4 \sigma_x P_{r,x} + u_x P_r -2 \sigma_{xx} F_r (u) -2 \sigma_{xx}P_r.
\end{gather*}
Since $2\sigma_{x,t_r} =u_{t_r} - \widetilde{u}_{t_r}$, we have
\[ 2\sigma_{x,t_r} = \dfrac{1}{2} P_{r,xxx} + 2 E P_{r,x} - 2 u P_{r,x} +4 \sigma_x F_{r,x} (u) +4 \sigma_x P_{r,x} - u_x P_r +2 \sigma_{xx} F_r (u) +2 \sigma_{xx}P_r.\]
Applying (2) of Proposition~\ref{cor: darboux fi} to the expresion $\sigma_x P_{r,x}$, we find
\begin{gather*}
 2\sigma_{x,t_r} = \dfrac{1}{2} P_{r,xxx} + 2 E P_{r,x} - 2 u P_{r,x} +4 \sigma_x F_{r,x} (u) + 3 \sigma_x P_{r,x} + \sigma_x (-2 \sigma P_r - 2 F_{r,x} (u)) \\
\hphantom{2\sigma_{x,t_r} =}{} - u_x P_r +2 \sigma_{xx} F_r (u) +2 \sigma_{xx}P_r \\
\hphantom{2\sigma_{x,t_r}}{}= \dfrac{1}{2} P_{r,xxx} + 2 E P_{r,x} +2 (\sigma_{xx} F_r (u) + \sigma_x F_{r,x} (u)) + P_{r,x} (-2u + 3 \sigma_x) \\
\hphantom{2\sigma_{x,t_r} =}{} + P_r (-2 \sigma \sigma_x -u_x +2 \sigma_{xx}).
\end{gather*}
Moreover, for the coefficient of $P_r$ we have
\[ -2 \sigma \sigma_x -u_x +2 \sigma_{xx} = \big({-}\sigma^2 -u +2 \sigma_x \big)_x = (-2u +3 \sigma_x )_x \]
by \eqref{eq:riccati sigma}. Thus, we obtain
\[ 2\sigma_{x,t_r} = \left ( \dfrac{1}{2} P_{r,xx} + 2 E P_{r} +2 \sigma_{x} F_r (u) + P_r (-2u + 3 \sigma_x) \right )_x.\]
Therefore, we have proved the statement.
\end{proof}

We finish this section with the following technical result. It makes a connection between differential polynomials $f_r (u)$ and some differential polynomials~$g_r (\sigma) $ defined by
\begin{gather}\label{eq:tr sigma}
 g_r (\sigma) := -A_{r+1} = \dfrac{1}{2} P_{r,xx} + 2 E P_{r} +2 \sigma_{x} F_r (u) + P_r (-2u + 3 \sigma_x).
\end{gather}

\begin{Proposition}\label{prop:gr y fr}We have the following relations:
\begin{enumerate}\itemsep=0pt
 \item[$1)$] $ (2 \sigma +\partial_x ) g_r (\sigma) = 2 f_{r+1, x} (u) = -\frac{1}{2} F_{r, xxx} (u) +2 (u-E)F_{r,x}(u) +u_x F_r (u)$ and
 \item[$2)$] $ (2 \sigma - \partial_x) g_r (\sigma) = 2 f_{r+1, x} (\widetilde{u}) = -\frac{1}{2} F_{r, xxx}(\widetilde{u}) +2 (\widetilde{u} -E)F_{r,x}(\widetilde{u}) +\widetilde{u}_x F_r(\widetilde{u}) $.
\end{enumerate}
\end{Proposition}

\begin{proof}The statement 1 is the statement 2 of Theorem~\ref{lemma: darboux fi} rewritten. For statement 2 we have
\begin{align*}
2 f_{r+1, x} (\widetilde{u}) & = 2 f_{r+1, x} (u) + 2A_{r+1,x} = 2\sigma g_r (\sigma) + g_{r,x} (\sigma) - 2 g_{r,x} (\sigma) = 2\sigma g_r (\sigma) - g_{r,x} (\sigma) \\
& = (2\sigma - \partial_x) g_r (\sigma)
\end{align*}
by statement 1 and equation \eqref{eq:tr sigma}.
 \end{proof}

\section[Fundamental matrices for KdV$_r$ rational Schr\"odinger operators]{Fundamental matrices for KdV$\boldsymbol{{}_r}$ rational\\ Schr\"odinger operators}\label{sec-fundamental}

In this section we obtain a fundamental matrix for the system~\eqref{eq:systemkdv0} depending on the energy level~$E$. The spectral curve is the tool that will allow us to understand why fundamental matrices present different behaviours according to the values of the energy.

For stationary rational potentials $u^{(0)}_n =n(n+1)x^{-2}$, it is well known that the spectral curve associated to the following system
\begin{gather*}
\Phi_x = U^{(0)} \Phi = \begin{pmatrix} 0 & 1 \\ u^{(0)}_n -E & 0 \end{pmatrix} \Phi, \\
\Phi_{t_n} = V^{(0)}_n \Phi = \begin{pmatrix} G_n (u^{(0)}_n) & F_n (u^{(0)}_n) \\ -H_n (u^{(0)}_n) & -G_n (u^{(0)}_n) \end{pmatrix} \Phi
\end{gather*}
is the algebraic plane curve in $\nC^2$ given by
\begin{gather*}
 \Gamma_n \colon \ p_n (\mu,E)= \mu^2 - E^{2n+1} =0.
\end{gather*}
Whenever an Adler--Moser potential $u_{r,n}(x,t)$ is time dependent, we will consider $\Gamma_n$ as the spectral curve associated to its corresponding linear differential system~\eqref{eq:systemkdv0}. Observe that $(E,\mu)=(0,0)$ is the unique affine singular point of~$\Gamma_n$. It turns out that for $E\not=0$ the behaviour of the fundamental matrix associated to the system
\begin{gather}
\Phi_x = U \Phi = \begin{pmatrix} 0 & 1 \\ u_{r,n} -E & 0 \end{pmatrix} \Phi, \nonumber\\
\Phi_{t_r} = V_r \Phi = \begin{pmatrix} -\dfrac{F_{r,x} (u_{r,n})}{2} & F_r (u_{r,n}) \\ (u_{r,n} -E) F_r (u_{r,n}) - \dfrac{F_{r,xx} (u_{r,n})}{2} & \dfrac{F_{r,x} (u_{r,n})}{2} \end{pmatrix} \Phi \label{eq:systemkdv1E rn}
\end{gather}
presents the same {algebraic structure} since the point $P=(E,\mu)$ is a regular point of~$\Gamma_n$. A~fundamental matrix for $E = 0$ can be also computed. However, it is not obtained by a~specialization process from the fundamental matrix obtained for a regular point. We include some examples in this section.

\subsection[Fundamental matrices for $E=0$]{Fundamental matrices for $\boldsymbol{E=0}$}

In this section, we compute explicitly fundamental matrices of system~\eqref{eq:systemkdv0} when the potential~$u$ is $u_{r,n} = -2 (\log \theta_{r,n})_{xx}$ and $E=0$. Recall that $u_{r,n}$ is a solution of KdV$_r$. Hence, we study the system
 \begin{gather}
\Phi_x = U \Phi = \begin{pmatrix} 0 & 1 \\ u_{r,n} & 0 \end{pmatrix} \Phi, \nonumber\\
\Phi_{t_r} = V_r \Phi = \begin{pmatrix} -\dfrac{f_{r,x} (u_{r,n})}{2} & f_r (u_{r,n}) \\ u_{r,n} f_r (u_{r,n}) - \dfrac{f_{r,xx} (u_{r,n})}{2} & \dfrac{f_{r,x} (u_{r,n})}{2} \end{pmatrix} \Phi.\label{eq:systemkdvr}
\end{gather}
It is obvious that the zero curvature condition of this system is the KdV$_r$ equation for $c_i=0$, $i=1, \ldots,r$:
\begin{gather*}
 \partial_{t_r} (u_{r,n}) = 2 f_{r+1,x} (u_{r,n}).
\end{gather*}
From now on we will denote $u_{r,n, t_r} = \partial_{t_r} (u_{r,n})$.

We have the following result:
\begin{Theorem}\label{soluciones1} Let $n$ be a non negative integer. For $E =0$ and $ u=u_{r,n}$, a fundamental matrix for system~\eqref{eq:systemkdvr} is
\begin{gather*}
\mathcal{B}^{(r)}_{n,0} = \begin{pmatrix} \phi_{1,r,n} & \phi_{2,r,n} \\ \phi_{1,r,n,x} & \phi_{2,r,n,x} \end{pmatrix},
\end{gather*}
where
\begin{gather*}
\phi_{1,r,n} (x, t_r,0)= \dfrac{\theta_{r,n-1}}{\theta_{r,n}} \qquad \textrm{and} \qquad \phi_{2,r,n} (x, t_r,0)= \dfrac{\theta_{r,n+1}}{\theta_{r,n}}.
\end{gather*}
\end{Theorem}
For $n=0$ we define $\theta_{r,-1} :=1$. We notice that $\phi_{2,r,n} =(\phi_{1,r,n+1})^{-1}$.

\begin{proof} We prove it by induction on $n$.
For $n=0$ the definition $\theta_{r,0} = 1$ gives $u_{r,0} =0$. So, the system \eqref{eq:systemkdvr} reads
 \begin{gather*}
\begin{pmatrix} \phi_{1,r,0,x} & \phi_{2,r,0,x} \\ \phi_{1,r,0,xx} & \phi_{2,r,0,xx} \end{pmatrix} = \begin{pmatrix} 0 & 1 \\ 0 & 0 \end{pmatrix} \begin{pmatrix} \phi_{1,r,0} & \phi_{2,r,0} \\ \phi_{1,r,0,x} & \phi_{2,r,0,x} \end{pmatrix} = \begin{pmatrix} \phi_{1,r,0,x} & \phi_{2,r,0,x} \\ 0 & 0 \end{pmatrix}, \\
\begin{pmatrix} \phi_{1,r,0,t_r} & \phi_{2,r,0,t_r} \\ \phi_{1,r,0,xt_r} & \phi_{2,r0,xt_r} \end{pmatrix} = \begin{pmatrix} 0 & 0 \\ 0 & 0 \end{pmatrix} \begin{pmatrix} \phi_{1,r,0} & \phi_{2,r,0} \\ \phi_{1,r,0,x} & \phi_{2,r,0,x} \end{pmatrix} = \begin{pmatrix} 0 & 0 \\ 0 & 0 \end{pmatrix}.
\end{gather*}
Thus, $\phi_{1,r,0} = 1$ and $\phi_{2,r,0} = x$ generate $\mathcal{B}^{(r)}_{0,0}$. Since $\theta_{r,1} = x$ we have that $\phi_{1,r,0} = \frac{\theta_{r,-1}}{\theta_{r,0}}$ and $\phi_{2,r,0} = \frac{\theta_{r,1}}{\theta_{r,0}}$.

Now, we suppose the statement is true for $n$ and prove it for $n+1$. For $n$ we know that $\phi_{1,r,n} = \frac{\theta_{r,n-1}}{\theta_{r,n}}$ and $ \phi_{2,r,n} = \frac{\theta_{r,n+1}}{\theta_{r,n}}$ generate $\mathcal{B}^{(r)}_{n,0}$. Therefore, $\phi_{1,r,n}$ and $\phi_{2,r,n}$ are solutions of Schr\"odinger equation $\phi_{xx} = u_{r,n} \phi$. We apply a Darboux transformation with $\phi_{2,r,n}$ to this Schr\"odinger equation and we obtain
\begin{gather}
{\rm DT}(\phi_{2,r,n}) {u_{r,n}} = u_{r,n} - 2 (\log \phi_{2,r,n})_{xx} = -2 (\log \theta_{r,n})_{xx} - 2 (\log \phi_{2,r,n})_{xx} \nonumber\\
\hphantom{{\rm DT}(\phi_{2,r,n}) {u_{r,n}}}{} = -2 (\log \phi_{2,r,n}\theta_{r,n} )_{xx} = -2 (\log \theta_{r,n+1})_{xx} = u_{r,n+1}, \label{eq:crum unr}\\
\label{eq:crum phi1r} {\rm DT}(\phi_{2,r,n})\phi_{1,r,n} = \phi_{1,r,n,x} - \dfrac{\phi_{2,r,n,x}}{\phi_{2,r,n}} \phi_{1,r,n} = -(2n+1)\dfrac{\theta_{n}}{\theta_{r,n+1}} = -(2n+1) \phi_{1,r,n+1}.
\end{gather}
So, $\phi_{1,r,n+1}= \frac{\theta_{r,n}}{\theta_{r,n+1}}$ is a solution of $\phi_{xx} =u_{r,n+1} \phi$ and, obviously, $ ( \phi_{1,r,n+1}, \phi_{1,r,n+1,x} )^t$ is a~column solution of the first equation of the system for $u_{r,n+1}$.

Now we verify that this column matrix is also a solution of the second equation
\begin{align*}
\begin{pmatrix} \phi_{1,r,n+1,t_r} \\ \phi_{1,r,n+1,xt_r} \end{pmatrix} &= \begin{pmatrix} -\dfrac{f_{r,x} (u_{r,n+1})}{2} & f_r (u_{r,n+1}) \\ u_{r,n+1} f_r (u_{r,n+1}) - \dfrac{f_{r,xx} (u_{r,n+1})}{2} & \dfrac{f_{r,x} (u_{r,n+1})}{2} \end{pmatrix} \begin{pmatrix} \phi_{1,r,n+1} \\ \phi_{1,r,n+1,x} \end{pmatrix} \\
& = \begin{pmatrix} -\dfrac{f_{r,x} (u_{r,n+1})}{2} \phi_{1,r,n+1} + f_r (u_{r,n+1}) \phi_{1,r,n+1,x} \\ \left(u_{r,n+1} f_r (u_{r,n+1}) - \dfrac{f_{r,xx} (u_{r,n+1})}{2} \right) \phi_{1,r,n+1} + \dfrac{f_{r,x} (u_{r,n+1})}{2} \phi_{1,r,n+1,x} \end{pmatrix}.
\end{align*}
We notice that the second row is just the partial derivative with respect to $x$ of the first one. Hence, we just have to verify that expressions \eqref{eq:crum unr} and \eqref{eq:crum phi1r} satisfy the equation
\begin{gather}\label{eq:chapter4 phirn+1}
 \phi_{1,r,n+1,t_r} = -\dfrac{f_{r,x} (u_{r,n+1})}{2} \phi_{1,r,n+1} + f_r (u_{r,n+1}) \phi_{1,r,n+1,x}.
\end{gather}

Applying expression \eqref{eq:crum phi1r} and the induction hypothesis we obtain for the left hand side of this equation
\begin{gather}
 \label{eq:chapter4 lefthandside} \phi_{1,r,n+1,t_r} = \dfrac{1}{2n+1} \left ( \phi_{1,r,n}\dfrac{\phi_{2,r,n,x}}{\phi_{2,r,n}} -\phi_{1,r,n,x} \right ) \left( \dfrac{f_{r,x} (u_{r,n})}{2} - f_r (u_{r,n}) \dfrac{\phi_{2,r,n,x}}{\phi_{2,r,n}} \right),
\end{gather}
and for the right hand side
\begin{gather}
 - \dfrac{f_{r,x} (u_{r,n+1})}{2} \phi_{1,r,n+1} + f_r (u_{r,n+1}) \phi_{1,r,n+1,x} \nonumber \\
 \qquad{} = \dfrac{1}{2n+1} \left ( \phi_{1,r,n}\dfrac{\phi_{2,r,n,x}}{\phi_{2,r,n}} -\phi_{1,r,n,x} \right ) \cdot \left( - \dfrac{f_{r,x} (u_{r,n+1})}{2} - f_r (u_{r,n+1}) \dfrac{\phi_{2,r,n,x}}{\phi_{2,r,n}} \right).\label{eq:chapter4 righthandside}
\end{gather}
Now, we prove that both expressions are equal. By applying the statement 2 of Theorem \ref{lemma: darboux fi} for $\sigma = \frac{\phi_{2,r,n,x}}{\phi_{2,r,n}} $, expression \eqref{eq:chapter4 righthandside} turns into{\samepage
\begin{gather*}
 - \dfrac{f_{r,x} (u_{r,n+1})}{2} - f_r (u_{r,n+1}) \dfrac{\phi_{2,r,n,x}}{\phi_{2,r,n}} = - \dfrac{f_{r,x} (u_{r,n}) + A_{r,x}}{2} - ( f_r (u_{r,n}) +A_r ) \dfrac{\phi_{2,r,n,x}}{\phi_{2,r,n}} \\
\qquad{} = - \dfrac{f_{r,x}(u_{r,n})}{2} - f_r (u_{r,n}) \dfrac{\phi_{2,r,n,x}}{\phi_{2,r,n}} - \dfrac{A_{r,x}}{2} - A_r \dfrac{\phi_{2,r,n,x}}{\phi_{2,r,n}} \\
\qquad{} = - \dfrac{f_{r,x}(u_{r,n})}{2} - f_r (u_{r,n}) \dfrac{\phi_{2,r,n,x}}{\phi_{2,r,n}} - \dfrac{A_{r,x}}{2} + f_{r,x} (u_{r,n}) + \dfrac{A_{r,x}}{2} \\
 \qquad{} = \dfrac{f_{r,x}(u_{r,n})}{2} - f_r (u_{r,n})\dfrac{\phi_{2,r,n,x}}{\phi_{2,r,n}},
\end{gather*}
which is equal to expression \eqref{eq:chapter4 lefthandside}. Therefore, both sides of expression \eqref{eq:chapter4 phirn+1} coincide.}

Now we proceed as in \cite{AM}. We take another column solution $( \phi_{2,r,n+1}, \phi_{2,r,n+1,x} )^t$ of this system for potential $u_{r,n+1}$ which is linearly independent of the one we have just computed, i.e., $\det \mathcal{B}^{(r)}_{n+1,0} $ is a nontrivial constant. We take $\phi_{2,r,n+1}$ such that
\[ \det \mathcal{B}^{(r)}_{n+1,0} =2(n+1)+1. \]
We notice that with this condition we have
\[ \det \mathcal{B}^{(r)}_{n+1,0} = \phi_{2,r,n+1,x} \dfrac{\theta_{r,n}}{\theta_{r,n+1}} - \phi_{2, r,n+1} \dfrac{\theta_{r,n,x} \theta_{r,n+1} - \theta_{r,n} \theta_{r,n+1,x}}{\theta_{r,n+1}^2} = 2(n+1)+1,\]
multiplying both sides by $\theta_{r,n+1}^2$ and using the recursion formula \eqref{eq:rec_dif} we get
\[ \phi_{2,r,n+1,x} \theta_{r,n} \theta_{r,n+1} - \phi_{2,r, n+1}( \theta_{r,n,x} \theta_{r,n+1} - \theta_{r,n} \theta_{r,n+1,x}) = \theta_{r,n+2,x} \theta_{r,n} - \theta_{r,n+2} \theta_{r,n,x}. \]
Setting $\phi_{2,r,n+1}= \dfrac{\alpha_{2,r,n+1}}{\theta_{r,n+1}}$ yields to
\[ \alpha_{2,r,n+1,x} \theta_{r,n} - \alpha_{2,r,n+1} \theta_{r,n,x} = \theta_{r,n+2,x} \theta_{r,n} - \theta_{r,n+2} \theta_{r,n,x}, \]
thus, $\alpha_{2,r,n+1} = \theta_{r,n+2} $ and $\phi_{2,r,n+1} = \frac{\theta_{r,n+2}}{\theta_{r,n+1}}$. This concludes the proof.
\end{proof}

Adler and Moser proved in \cite{AM} that matrix $\mathcal{B}^{(r)}_{n,0}$ is a fundamental matrix for the Schr\"odinger equation~\eqref{eq:schr0} for $E=0$. But they did not prove there that this matrix is also a fundamental matrix for the second equation of the system~\eqref{eq:systemkdvr}. To do that, it is necessary to control the action of the Darboux transformations over the differential polynomials $f_j$, as we did in Section~\ref{sect-novikovs DT}.

\begin{Remark} Since $\phi_{1,r,n} = \frac{\theta_{r,n-1}}{\theta_{r,n}}$ and $\phi_{2,r,n} = \frac{\theta_{r,n+1}}{\theta_{r,n}}$ are solutions of Schr\"odinger equation~\eqref{eq:schr0} for~$E=0$, this translate into the following equation for polynomials $\theta_{r,n}$:
\begin{gather}\label{eq:rec2}
 \theta_{r,n+1,xx} \theta_{r,n} + \theta_{r,n+1} \theta_{r,n,xx} - 2\theta_{r,n,x} \theta_{r,n+1,x} =0.
\end{gather}
\end{Remark}

\begin{Theorem}We have that
\begin{gather*}
 \det \mathcal{B}^{(r)}_{n,0} = 2n+1.
\end{gather*}
\end{Theorem}

\begin{Example}\label{ex:E=0}To illustrate the results, we present explicit computations using SAGE of fundamental solutions of the system for the first values of $n$.

1. First, we show the first examples of unadjusted fundamental solutions
 \begin{gather*}\label{eq:tabla E=0 taus}
 \begin{matrix} n & \phi_{1,r,n} & \phi_{2,r,n} & u_{r,n} \\[5pt]
 0 & 1 & x & 0 \\[3pt]
 1 & \dfrac{1}{x} & \dfrac{x^3 + \tau_2}{x} & \dfrac{2}{x^2} \\[13pt]
 2 & \dfrac{x}{x^3 + \tau_2} & \dfrac{x^6 + 5x^3 \tau_2 + x \tau_3 -5 \tau_2^2}{x^3 + \tau_2} & \dfrac{6x(x^3 -2\tau_2 )}{(x^3 + \tau_2)^2} \\[13pt]
 3 & \dfrac{x^3 + \tau_2}{x^6 + 5x^3 \tau_2 + x \tau_3 -5 \tau_2^2} & \dfrac{p_1 (x, \tau_2, \tau_3, \tau_4)}{x^6 + 5x^3 \tau_2 + x \tau_3 -5 \tau_2^2} & \dfrac{p_2 (x, \tau_2, \tau_3)}{(x^6 + 5x^3 \tau_2 + x \tau_3 -5 \tau_2^2)^2} \end{matrix}
\end{gather*}
where
\begin{gather*}
\begin{split}& p_1 (x, \tau_2, \tau_3, \tau_4) = x^{10} + 15 x^7 \tau_2 +7 x^5 \tau_3 -35 x^2 \tau_2 \tau_3 +175 x \tau_2^3 - \dfrac{7}{3} \tau_3^2 + x^3 \tau_4 + \tau_2 \tau_4,\\
& p_2 (x, \tau_2, \tau_3) = 12x^{10} -36 x^5 \tau_3 +450 x^4 \tau_2^2 + 300x \tau_2^3 +2\tau_3^2.
\end{split}
\end{gather*}

2. Next, we compute fundamental solutions for potentials which are solutions of the first level of the KdV hierarchy, KdV$_1$ equation: $ u_{t_1} = \frac{3}{2} uu_x - \frac{1}{4} u_{xxx}$. We also show the explicit choice of the functions $\tau_i$
\begin{gather*}\label{eq:tabla E=0}
 \begin{matrix} n & \phi_{1,1,n} & \phi_{2,1,n} & u_{1,n} & \!(\tau_2, \ldots,\tau_{n})\! \\[5pt]
 0 & 1 & x & 0 & \\[3pt]
 1 & \dfrac{1}{x} & \dfrac{x^3 +3t_1}{x} & \dfrac{2}{x^2} & (3t_1) \\[13pt]
 2 & \dfrac{x}{x^3 + 3t_1} & \dfrac{x^6 + 15x^3t_1 -45t_1^2}{x^3 + 3t_1} & \dfrac{6x(x^3 -6t_1)}{(x^3 +3t_1)^2} & (3t_1,0) \\[13pt]
 3\! & \!\dfrac{x^3 + 3t_1}{x^6 + 15x^3t_1 -45t_1^2}\! & \!\dfrac{x^{10} +45x^7 t_1 + 4725xt_1^3}{x^6 + 15x^3t_1 -45t_1^2}\! & \!\dfrac{6x \big(2x^9 + 675x^3t_1^2 +1350t_1^3\big)}{\big(x^6 + 15x^3t_1 -45t_1^2\big)^2}\! & \!(3t_1,0,0)\! \end{matrix}
\end{gather*}
\end{Example}

\subsection[Fundamental matrices for $E\neq0$]{Fundamental matrices for $\boldsymbol{E\neq0}$}

In this section, we compute explicitly fundamental matrices of system \eqref{eq:systemkdv0} when $u=u_{r,n} = -2 (\log \theta_{r,n})_{xx}$ and $E \neq 0$. In this case, the system is
 \begin{gather}
\Phi_x = U \Phi = \begin{pmatrix} 0 & 1 \\ u_{r,n} -E & 0 \end{pmatrix} \Phi, \nonumber\\
\Phi_{t_r} = V_r \Phi = \begin{pmatrix} -\dfrac{F_{r,x} (u_{r,n})}{2} & F_r (u_{r,n}) \\ (u_{r,n} -E) F_r (u_{r,n}) - \dfrac{F_{r,xx} (u_{r,n})}{2} & \dfrac{F_{r,x} (u_{r,n})}{2} \end{pmatrix} \Phi.\label{eq:systemkdv1E r}
\end{gather}
The zero curvature condition of this system is still the KdV$_r$ equation for $c_i=0$, $i=1, \ldots,r$:
\begin{gather*}
 u_{r,n, t_r} = 2 f_{r+1,x} (u_{r,n}).
\end{gather*}
When $E \neq 0$, we take $\lambda \in \nC$ a parameter over $K$ such that $E+ \lambda^2 =0$.

Next, we consider the differential systems
\begin{gather}
\label{eq:q+1,x r} Q^+_{n,xx} = Q^+_{n,x} \left (-2 \lambda + 2 \dfrac{\theta_{r,n,x}}{\theta_{r,n}} \right ) + Q^+_{n} \left ( 2 \lambda \dfrac{\theta_{r,n,x}}{\theta_{r,n}} -\dfrac{\theta_{r,n,xx}}{\theta_{r,n}} \right ), \\
Q^+_{n,t_r} = Q^+_{n,x} F_r (u_{r,n}) \nonumber \\
\hphantom{Q^+_{n,t_r} =}{} + Q^+_{n} \left( -(-1)^r \lambda^{2r+1} + \lambda F_r (u_{r,n}) + \dfrac{\theta_{r,n,t_r}}{\theta_{r,n}} - \dfrac{F_{r,x} (u_{r,n})}{2} - F_r (u_{r,n}) \dfrac{\theta_{r,n,x}}{\theta_{r,n}} \right), \label{eq:q+1,t r}\\
\label{eq:q-1,x r} Q^-_{n,xx} = Q^-_{n,x} \left (2 \lambda + 2 \dfrac{\theta_{r,n,x}}{\theta_{r,n}} \right ) - Q^-_{n} \left ( 2 \lambda \dfrac{\theta_{r,n,x}}{\theta_{r,n}} + \dfrac{\theta_{r,n,xx}}{\theta_{r,n}} \right ), \\
Q^-_{n,t_r} = Q^-_{n,x} F_r (u_{r,n}) \nonumber \\
\hphantom{Q^-_{n,t_r} =}{} +Q^-_{n} \left( (-1)^r \lambda^{2r+1} - \lambda F_r (u_{r,n}) + \dfrac{\theta_{r,n,t_r}}{\theta_{r,n}} - \dfrac{F_{r,x} (u_{r,n})}{2} - F_r (u_{r,n}) \dfrac{\theta_{r,n,x}}{\theta_{r,n}} \right).\label{eq:q-1,t r}
\end{gather}

We have the following relations for the solutions of the differential systems \eqref{eq:q+1,x r}--\eqref{eq:q+1,t r} and \eqref{eq:q-1,x r}--\eqref{eq:q-1,t r}.

\begin{Lemma} \label{lem-Qn} Functions $Q^+_{n}$ and $Q^-_{n}$ recursively defined by
\begin{gather}
\label{eq:q+2 r}
Q^+_{0} =1, \qquad Q^+_{n+1} = \dfrac{\lambda Q^+_{n} \theta_{r,n+1} + Q^+_{n,x} \theta_{r,n+1} - Q^+_{n} \theta_{r,n+1, x}}{\theta_{r,n}}, \\
\label{eq:q-2 r}
Q^-_0 =1, \qquad Q^-_{n+1} = \dfrac{\lambda Q^-_{n} \theta_{r,n+1} - Q^-_{n,x} \theta_{r,n+1} + Q^-_{n} \theta_{r,n+1, x}}{\theta_{r,n}}
\end{gather}
are solutions of the differential systems \eqref{eq:q+1,x r}--\eqref{eq:q+1,t r} and \eqref{eq:q-1,x r}--\eqref{eq:q-1,t r}.
\end{Lemma}

\begin{proof}
We prove it by induction on $n$. For $n=0$ we have $\theta_{r,0} =1$, hence, $u_{r,0} =0$ and $F_r (u_{r,0}) = (-1)^r \lambda^{2r}$. So, $Q^+_0 =1$ and $Q^-_0 =1$ are solutions of the systems \eqref{eq:q+1,x r}--\eqref{eq:q+1,t r} and \eqref{eq:q-1,x r}--\eqref{eq:q-1,t r}.

Now, we suppose it is true for $n$ and prove it for $n+1$. We have to prove that expressions
\begin{gather*}
 Q^+_{n+1} = \dfrac{\lambda Q^+_{n} \theta_{r,n+1} + Q^+_{n,x} \theta_{r,n+1} - Q^+_{n} \theta_{r,n+1, x}}{\theta_{r,n}},\\
 Q^-_{n+1} = \dfrac{\lambda Q^-_{n} \theta_{r,n+1} - Q^-_{n,x} \theta_{r,n+1} + Q^-_{n} \theta_{r,n+1, x}}{\theta_{r,n}}
\end{gather*}
satisfy equations \eqref{eq:q+1,x r}, \eqref{eq:q+1,t r}, \eqref{eq:q-1,x r} and \eqref{eq:q-1,t r} respectively, for $n+1$. First, we prove that $Q^+_{n+1}$ satisfies \eqref{eq:q+1,x r} and \eqref{eq:q+1,t r}. By induction hypothesis, we know that $Q^+_{n}$ satisfies \eqref{eq:q+1,x r}, using this expression and \eqref{eq:rec2} we have
\begin{gather*}
Q^+_{n+1,x} = \dfrac{\lambda Q^+_{n} \theta_{r,n+1,x} -\lambda Q^+_{n,x} \theta_{r,n+1}}{\theta_{r,n}} + \dfrac{(\lambda Q^+_{n} \theta_{r,n+1} + Q^+_{n,x} \theta_{r,n+1} - Q^+_{n} \theta_{r,n+1,x}) \theta_{r,n,x}}{\theta_{r,n}^2},\\
Q^+_{n+1,xx} = \dfrac{Q^+_{n,x}}{\theta_{r,n}^3} p_1 (x,t_r, \lambda) + \dfrac{Q^+_{n}}{\theta_{r,n}^3} p_2 (x,t_r, \lambda),
\end{gather*}
and
\begin{gather*}
 Q^+_{n+1,x} \left( -2 \lambda + 2 \dfrac{\theta_{r,n+1,x}}{\theta_{r,n+1}} \right) + Q^+_{n+1} \left ( 2 \lambda \dfrac{\theta_{r,n+1,x}}{\theta_{r,n+1}} -\dfrac{\theta_{r,n+1,xx}}{\theta_{r,n+1}} \right)\\
 \qquad{} = \dfrac{Q^+_{n,x}}{\theta_{r,n}^3} p_1 (x,t_r, \lambda) + \dfrac{Q^+_{n}}{\theta_{r,n}^3} p_2 (x,t_r, \lambda),
\end{gather*}
where
\begin{gather*}
p_1 (x,t_r, \lambda) = 2\lambda^2 \theta^2_{r,n} \theta_{r,n+1} - 2\lambda \theta_{r,n}\theta_{r,n,x} \theta_{r,n+1} + 2\theta_{r,n} \theta_{r,n,x} \theta_{r,n+1,x} - \theta_{r,n}^2 \theta_{r,n+1,xx},\\
p_2 (x,t_r,\lambda) = - 2\lambda^2 \theta_{r,n}\theta_{r,n,x} \theta_{r,n+1} + 2\lambda \theta_{r,n}\theta_{r,n,xx} \theta_{r,n+1} + \theta_{r,n}^2 \theta_{r,n+1,xx} - \theta_{r,n}\theta_{r,n,xx} \theta_{r,n+1,x}.
\end{gather*}

Thus, both expressions coincide and $Q^+_{n+1}$ is solution of equation \eqref{eq:q+1,x r}.

On the other hand, by induction hypothesis, we know that $Q^+_{n}$ satisfies \eqref{eq:q+1,t r}. Using this equation, expressions
\begin{gather*}
\sigma_{2,r,n} =( \log \phi_{2,r,n})_x = \dfrac{\theta_{r,n+1,x} \theta_{r,n} - \theta_{r,n+1} \theta_{r,n,x}}{\theta_{r,n} \theta_{r,n+1}}, \\
\sigma_{2,r,n, t_r} = \dfrac{\theta_{r,n+1,x t_r}}{\theta_{r,n+1}} - \dfrac{\theta_{r,n,xt_r}}{\theta_{r,n}} + \dfrac{\theta_{r,n,x} \theta_{r,n,t_r}}{\theta_{r,n}^2} - \dfrac{\theta_{r,n+1,x} \theta_{r,n+1, t_r}}{\theta_{r,n+1}^2}, \\
Q^+_{n,xt_r} = Q^+_{n,x} \left( -(-1)^r \lambda^{2r+1} - \lambda F_r (u_{r,n}) + \dfrac{F_{r,x} (u_{r,n})}{2} + F_r (u_{r,n}) \dfrac{\theta_{r,n,x}}{\theta_{r,n}} + \dfrac{\theta_{r,n,t_r}}{\theta_{r,n}} \right) \\
\hphantom{Q^+_{n,xt_r} =}{} + Q^+_{n} \left( 2 \lambda F_r (u_{r,n}) \dfrac{\theta_{r,n,x}}{\theta_{r,n}} + \lambda F_{r,x} (u_{r,n}) -2 F_r (u_{r,n}) \dfrac{\theta_{r,n,xx}}{\theta_{r,n}} + F_r (u_{r,n}) \dfrac{\theta^2_{r,n,x}}{\theta_{r,n}^2} \right. \\
 \left. \hphantom{Q^+_{n,xt_r} =}{} - \dfrac{F_{r,xx} (u_{r,n})}{2} - F_{r,x} (u_{r,n}) \dfrac{\theta_{r,n,x}}{\theta_{r,n}} - \dfrac{\theta_{r,n,x} \theta_{r,n,t_r}}{\theta^2_{r,n}} + \dfrac{\theta_{r,n,xt_r}}{\theta_{r,n}} \right),
\end{gather*}
the derivative with respect to $x$ of statement~2 of Corollary~\ref{cor: darboux fi} and expression~\eqref{eq:chapter3 sigma_tr} for $\sigma_{2,r,n,t_r}$, we obtain
\[ Q^+_{n+1,t_r} = Q^+_{n,x} \dfrac{p_3 (x,t_r, \lambda)}{\theta_{r,n}^2} + Q^+_{n} \dfrac{p_4 (x,t_r, \lambda)}{\theta_{r,n}^2}, \]
where
\begin{gather*}
 p_3 (x,t_r, \lambda) = -(-1)^r \lambda^{2r+1} \theta_{r,n} \theta_{r,n+1} + F_r (u_{r,n}) \theta_{r,n,x} \theta_{r,n+1} - F_r (u_{r,n}) \theta_{r,n} \theta_{r,n+1,x} \\
\hphantom{p_3 (x,t_r, \lambda) =}{} + F_{r,x} (u_{r,n}) \dfrac{\theta_{r,n} \theta_{r,n+1}}{2} + \theta_{r,n} \theta_{r,n+1,t_r},\\
 p_4 (x,t_r, \lambda) = -(-1)^r \lambda^{2r+2} \theta_{r,n} \theta_{r,n+1} +(-1)^r \lambda^{2r+1} \theta_{r,n} \theta_{r,n+1,x} + \lambda^2 F_r (u_{r,n}) \theta_{r,n} \theta_{r,n+1} \\
\hphantom{p_4 (x,t_r, \lambda) =}{} + \lambda^2 P_r \theta_{r,n} \theta_{r,n+1} + \lambda \theta_{r,n} \theta_{r,n+1,t_r} + \lambda F_r (u_{r,n}) \theta_{r,n,x} \theta_{r,n+1} \\
\hphantom{p_4 (x,t_r, \lambda) =}{} + \lambda F_{r,x} (u_{r,n}) \dfrac{\theta_{r,n} \theta_{r,n+1}}{2}- \lambda F_r (u_{r,n}) \theta_{r,n} \theta_{r,n+1,x} + F_{r,x} (u_{r,n}) \dfrac{\theta_{r,n} \theta_{r,n+1,x}}{2} \\
\hphantom{p_4 (x,t_r, \lambda) =}{} - P_r \theta_{r,n,x} \theta_{r,n+1,x}- \dfrac{\theta_{r,n} \theta_{r,n+1,x} \theta_{r,n+1, t_r}}{\theta_{r,n+1}} - F_r (u_{r,n}) \theta_{r,n,x} \theta_{r,n+1,x} \\
\hphantom{p_4 (x,t_r, \lambda) =}{} + F_r (u_{r,n}) \dfrac{\theta_{r,n} \theta_{r,n+1,x}^2}{\theta_{r,n+1}} + P_r \dfrac{\theta_{r,n} \theta_{r,n+1,x}^2}{\theta_{r,n+1}} + P_{r,x} \dfrac{\theta_{r,n} \theta_{r,n+1,x}}{2}.
\end{gather*}

Finally, using relation \eqref{eq:q+1,x r} for $Q^+_{n}$ and statements~1 and~2 of Corollary \ref{cor: darboux fi}, the right hand side of equation~\eqref{eq:q+1,t r} for $Q^+_{n+1}$ reads
\begin{gather*} Q^+_{n+1,x} F_r (u_{r,n+1}) + Q^+_{n+1} \left( \lambda^3 + \lambda F_r (u_{r,n+1}) + \dfrac{\theta_{r,n,t_r}}{\theta_{r,n}} - \dfrac{F_{r,x} (u_{r,n+1})}{2} - F_r (u_{r,n+1}) \dfrac{\theta_{r,n,x}}{\theta_{r,n}} \right) \\
\qquad{} = Q^+_{n,x} \dfrac{p_3 (x,t_r, \lambda)}{\theta_{r,n}^2} + Q^+_{n} \dfrac{p_4 (x,t_r, \lambda)}{\theta_{r,n}^2}.
\end{gather*}
Therefore, both expressions coincide and $Q^+_{n+1}$ is a solution of equation \eqref{eq:q+1,t r}.

The proof for $Q^-_{n+1}$ is analogous.
\end{proof}

As a consequence, we have the following result:

\begin{Theorem}\label{soluciones2 r}
Let $n$ be a non negative integer, then, for $E= -\lambda^2 \neq 0$ and $u=u_{r,n}$, a~fundamental matrix for system~\eqref{eq:systemkdv1E r} is
\begin{gather*}
\mathcal{B}^{(r)}_{n, \lambda}= \begin{pmatrix} \phi^+_{r,n} & \phi^-_{r,n} \\ \phi^+_{r,n,x} & \phi^-_{r,n,x} \end{pmatrix},
\end{gather*}
where
\begin{gather*}
\phi^+_{r,n} (x,t_r,\lambda)= e^{\lambda x +(-1)^r \lambda^{2r+1} t_r} \dfrac{Q^+_{r,n}(x,t_r, \lambda)}{\theta_{r,n}}, \\
 \phi^-_{r,n} (x,t_r,\lambda)= e^{-\lambda x - (-1)^r \lambda^{2r+1} t_r}\dfrac{Q^-_{r,n} (x,t_r, \lambda)}{\theta_{r,n}},
\end{gather*}
where $Q^+_{r,n}$ and $Q^-_{r,n}$ are functions in $x$, $t_r$, $\lambda$ defined by means of Lemma~{\rm \ref{lem-Qn}}.
\end{Theorem}

\begin{proof} We prove it by induction on $n$. For $n=0$ the definition $\theta_{r,0} =1$ leads to $u_{r,0}=0$. So, the system \eqref{eq:systemkdv1E r} becomes
\begin{gather*}
\begin{pmatrix} \phi^+_{r,0,x} & \phi^-_{r,0,x} \\ \phi^+_{r,0,xx} & \phi^-_{r,0,xx} \end{pmatrix} = \begin{pmatrix} 0 & 1 \\ \lambda^2 & 0 \end{pmatrix} \begin{pmatrix} \phi^+_{r,0} & \phi^-_{r,0} \\ \phi^+_{r,0,x} & \phi^-_{r,0,x} \end{pmatrix}, \\[10pt]
\begin{pmatrix}\phi^+_{r,0,t_r} & \phi^-_{r,0,t_r} \\ \phi^+_{r,0,xt_r} & \phi^-_{r,0,xt_r} \end{pmatrix}
 = \begin{pmatrix} 0 & (-1)^r \lambda^{2r} \\ (-1)^r \lambda^{2r+2} & 0 \end{pmatrix} \begin{pmatrix} \phi^+_{r,0} & \phi^-_{r,0} \\ \phi^+_{r,0,x} & \phi^-_{r,0,x} \end{pmatrix}.
\end{gather*}
Hence, $\phi^+_{r,0} = e^{\lambda x +(-1)^r \lambda^{2r+1} t_r}$ and $\phi^-_{r,0}= e^{-\lambda x -(-1)^r \lambda^{2r+1} t_r}$ generate $\mathcal{B}^{(r)}_{0,\lambda}$. Since $\theta_{r,0} =1$, we find $Q^\pm_{r,0} =1$, as in Lemma~\ref{lem-Qn}.

Next, we suppose it true for $n$ and prove it for $n+1$. Since
\begin{gather*}
\phi^+_{r,n} (x,t_r,\lambda)= e^{\lambda x +(-1)^r \lambda^{2r+1} t_r} \frac{Q^+_{r,n}}{\theta_{r,n}}, \qquad \phi^-_{r,n} (x,t_r,\lambda)= e^{-\lambda x -(-1)^r \lambda^{2r+1} t_r}\frac{Q^-_{r,n}}{\theta_{r,n}}
\end{gather*} are solutions of Schr\"odinger equation $\phi_{xx} =\big(u_{r,n} + \lambda^2\big) \phi$, we apply a Darboux transformation with $\phi_{2,r,n}= \frac{\theta_{r,n+1}}{\theta_{r,n}}$ to this equation and we obtain
\begin{gather}
{\rm DT}(\phi_{2,r,n}) {u_{r,n}} = u_{r,n} - 2 (\log \phi_{2,r,n})_{xx} = u_{r,n} - 2 \sigma_{2,r,n,x} = u_{r,n+1}, \nonumber \\
{\rm DT}(\phi_{2,r,n}){\phi}^+_{r,n} = \phi^+_{r,n,x} - \dfrac{\phi_{2,r,n,x}}{\phi_{2,r,n}} \phi^+_{r,n} \nonumber \\
\hphantom{{\rm DT}(\phi_{2,r,n}){\phi}^+_{r,n}}{}
= \dfrac{e^{\lambda x +(-1)^r \lambda^{2r+1} t_r}}{\theta_{r,n+1}} \cdot \dfrac{\lambda Q^+_{r,n} \theta_{r,n+1} + Q^+_{r,n,x} \theta_{r,n+1} - Q^+_{r,n} \theta_{r,n+1, x}}{\theta_{r,n}}\nonumber \\
\hphantom{{\rm DT}(\phi_{2,r,n}){\phi}^+_{r,n}}{} = e^{\lambda x +(-1)^r \lambda^{2r+1} t_r} \dfrac{Q^+_{r,n+1}}{\theta_{r,n+1}} = \phi^+_{r,n+1} (x,t_r,\lambda), \label{eq:DT phi+}\\
{\rm DT}(\phi_{2,r,n}){\phi}^-_{r,n} = \phi^-_{r,n,x} - \dfrac{\phi_{2,r,n,x}}{\phi_{2,r,n}} \phi^-_{r,n} \nonumber \\
\hphantom{{\rm DT}(\phi_{2,r,n}){\phi}^-_{r,n}}{} = \dfrac{e^{-\lambda x - (-1)^r \lambda^{2r+1} t_r}}{\theta_{r,n+1}} \cdot \dfrac{-\lambda Q^-_{r,n} \theta_{r,n+1} + Q^-_{r,n,x} \theta_{r,n+1} - Q^-_{r,n} \theta_{r,n+1, x}}{\theta_{r,n}}\nonumber \\
\hphantom{{\rm DT}(\phi_{2,r,n}){\phi}^-_{r,n}}{} = e^{-\lambda x -(-1)^r \lambda^{2r+1} t_r}\dfrac{(-Q^-_{r,n+1})}{\theta_{r,n+1}} = - \phi^-_{r,n+1} (x,t_r,\lambda),\label{eq:DT phi-}
\end{gather}
by Lemma \ref{lem-Qn}. Hence, ${\rm DT}(\phi_{2,r,n}){\phi}^+_{r,n} = \phi^+_{r,n+1} (x,t_r,\lambda) $ and $
 {\rm DT}(\phi_{2,r,n}){\phi}^-_{r,n} = - \phi^-_{r,n+1} (x,t_r,\lambda)$ generate $\mathcal{B}^{(r)}_{n+1,\lambda}$. This ends the proof.
\end{proof}

As far as we know, a general expression for fundamental matrices for system~\eqref{eq:systemkdv1E r} has never been computed when $E\neq 0$. {In the stationary case, i.e., in the case we only have the Schr\"odinger equation with Adler--Moser potentials, P.~Clarkson showed in~\cite{Clarkson} an expression for the fundamental solutions of this equation when $E\neq 0$. However these expressions are not explicit, so it is not convenient for studying the Galois groups.}

As in Theorem \ref{soluciones1}, the key to compute these solutions is to control the action of the Darboux transformations over the differential polynomials $f_j$, as we showed in Section~\ref{sect-novikovs DT}. In Section \ref{sect:Q+ Q-} we will give some examples of these fundamental solutions both in the general framework of unadjusted functions $\tau_i$ and in the particular case $r=1$, in the same line as in Example~\ref{ex:E=0}.

\begin{Proposition}\label{cor: relac phi+ y phi-}The functions $ Q^+_{r,n}$, $Q^-_{r,n}$ and the solutions $ \phi^+_{r,n}, \phi^-_{r,n}$ defined in Theorem~{\rm \ref{soluciones2 r}} satisfy the relations
\begin{gather*}
Q^+_{r,n} (x,t_r, - \lambda) = (-1)^n Q^-_{r,n} (x,t_r, \lambda) \qquad \text{and} \qquad \phi^+_{r,n} (x,t_r, - \lambda) = (-1)^n \phi^-_{r,n} (x,t_r, \lambda).
\end{gather*}
\end{Proposition}

\begin{proof}We notice that
\[ \phi^+_{r,n} (x,t_r, - \lambda) = e^{- \lambda x -(-1)^r \lambda^{2r+1} t_r}\dfrac{Q^+_{r,n}(x,t_r, -\lambda)}{\theta_{r,n}}, \]
since $\theta_{r,n}$ does not depend on $\lambda$. So, both relations are equivalent and it suffices to prove that $Q^+_{r,n} (x,t_r, - \lambda) = (-1)^n Q^-_{r,n} (x,t_r, \lambda) $. We prove it by induction on $n$.
For $n=0$, we have that $Q^+_{r,0} = 1 = Q^-_{r,0}$. Hence, $Q^+_{r,0} (x,t_r, - \lambda) = (-1)^0 Q^-_{r,0} (x,t_r, \lambda)$.

Using the expresions (\ref{eq:q+2 r}) and (\ref{eq:q-2 r}), we obtain
\begin{align*}
Q^+_{r,n+1} (x,t_r, -\lambda) &= \dfrac{(-\lambda \theta_{r,n+1} - \theta_{r,n+1, x}) Q^+_{r,n} (x,t_r, -\lambda) + Q^+_{r,n,x}(x,t_r, -\lambda) \theta_{r,n+1}}{\theta_{r,n}} \\
&= \dfrac{(-1)^n ((-\lambda \theta_{r,n+1} - \theta_{r,n+1, x} ) Q^-_{r,n} (x,t_r, \lambda) + Q^-_{r,n,x} (x,t_r, \lambda) \theta_{r,n+1} )}{\theta_{r,n}} \\
&= \dfrac{(-1)^{n+1} ((\lambda \theta_{r,n+1} + \theta_{r,n+1, x} ) Q^-_{r,n} (x,t_r, \lambda) -Q^-_{r,n,x} (x,t_r, \lambda) \theta_{r,n+1})}{\theta_{r,n}} \\
&= (-1)^{n+1} Q^-_{r,n+1} (x,t_r, \lambda),
\end{align*}
as we wanted to prove.
\end{proof}

This corollary allows us to compute the determinant of $ \mathcal{B}^{(r)}_{n,\lambda } $. First observe that
\begin{align}\label{eq:det Er}
\det \mathcal{B}^{(r)}_{n,\lambda } &= W( \phi^+_{r,n}, \phi^-_{r,n}) = (-1)^n W( \phi^+_{r,n} (x,t_r, \lambda), \phi^+_{r,n} (x,t_r, -\lambda)) \\
& = (-1)^{n+1} \dfrac{2 \lambda Q^+_{r,n} (x,t_r, \lambda) Q^+_{r,n} (x,t_r, -\lambda) + W (Q^+_{r,n} (x,t_r, -\lambda), Q^+_{r,n} (x,t_r, \lambda))}{\theta^2_{r,n}},\nonumber
\end{align}
where $W( \phi_1, \phi_2) = \phi_1 \phi_{2,x} - \phi_{1,x} \phi_{2} $ denotes the Wronskian of $\phi_1$ and $\phi_2$.

\begin{Theorem}\label{thm: det E} We have
\begin{gather*}
\det \mathcal{B}^{(r)}_{n,\lambda } = -2 \lambda^{2n+1}.
\end{gather*}
\end{Theorem}

\begin{proof} We proceed by induction on $n$. For $n=0$ we obtain $Q^+_{r,0} = 1$ and $\theta_{r,0} = 1$, so $\det \mathcal{B}_{0,\lambda}^{(r)}= -2 \lambda$. Now, we suppose it is true for $n$ and prove it for $n+1$. Replacing expression~\eqref{eq:q+2 r} for $Q^+_{r,n+1} (x,t_r, \lambda)$ and $Q^+_{r,n+1} (x,t_r, -\lambda)$ in formula~\eqref{eq:det Er} and using Proposition \ref{cor: relac phi+ y phi-} and the induction hypothesis, we get
\[ \det \mathcal{B}^{(r)}_{n+1,\lambda } = -2 \lambda^{2n+3}= -2 \lambda^{2(n+1) +1}. \]
As we wanted to prove.
\end{proof}

\begin{Remark}\label{rk:spec process}Theorem \ref{thm: det E} implies that the matrix $\mathcal{B}^{(r)}_{n,\lambda}$ is not a fundamental matrix of system~\eqref{eq:systemkdv0} for $\lambda =E=0$, since it is not invertible for that value of $E$. The reason of this is that, by Proposition~\ref{cor: relac phi+ y phi-}, when $\lambda =0$ we have $\phi^+_{r,n} (x,t_r, 0) = (-1)^{n} \phi^-_{r,n} (x,t_r, 0)$, so, both column solutions are linearly dependent. We will detail this phenomenon in Section~\ref{sec espec darb}. In fact, we will show that it is not the same to set $E=0$ in~\eqref{eq:systemkdv0} and then solve the system, than to solve the system for a generic $E$ and then replace $E=0$ in the solution obtained, i.e., there is not a~specialization process in this sense.
\end{Remark}

\begin{Example}\label{rem:tabla E no 0 taus}
For $n=0$ and $n=1$ we obtain by direct computations the following solutions:
\begin{gather*}
\begin{matrix}n & \phi^+_{r,n} & \phi^-_{r,n} \\[5pt]
 0 & e^{\lambda x +(-1)^r \lambda^{2r+1} t_r} & e^{-\lambda x -(-1)^r \lambda^{2r+1} t_r} \\[3pt]
 1 & \qquad e^{\lambda x +(-1)^r \lambda^{2r+1} t_r} \dfrac{\lambda x -1}{x} & \qquad e^{-\lambda x -(-1)^r \lambda^{2r+1} t_r} \dfrac{\lambda x +1}{x}
\end{matrix}
\end{gather*}
In next section we will show a method to compute functions $Q^+_{r,n}$ and $Q^-_{r,n}$ more efficient than solving explicitly equations \eqref{eq:q+1,x r}, \eqref{eq:q+1,t r}, \eqref{eq:q-1,x r} and \eqref{eq:q-1,t r}. This allow us to obtain fundamental matrices $\mathcal{B}^{(r)}_{n,\lambda } $. In particular $\phi^+_{r,1}$ and $\phi^-_{r,1}$ are linearly independent solutions for the Schr\"odinger operator $-\partial^2 +u_{r,1}-E=0$ where $u_{r,1}=2/x^2$ is the constructed rational KdV$_r$ potential, as long as $E\not=0$.
\end{Example}

\section[Examples of fundamental matrices for the case $E \neq 0$]{Examples of fundamental matrices for the case $\boldsymbol{E \neq 0}$}\label{sect:Q+ Q-}

Along this section we will prove that the funtions $Q^{\pm}_{r,n}$ defined in Theorem \ref{soluciones2 r} satisfy the recursion formula \eqref{eq:rec_dif}. This implies in particular that they are polynomials of $x$ with coefficients in $\nC(\lambda, t_r)$. Thus, they generalize the family of Adler--Moser polynomials~$\theta_n$.

For the following computations we do not suppose that functions $\theta_n$ and $Q^{\pm}_n$ and poten\-tials~$u_n$ are adjusted to any level of the KdV hierarchy.

\subsection{Generalized Adler--Moser polynomials}

In Lemma \ref{lem-Qn} we have obtained the recursive formulas \eqref{eq:q+2 r} and \eqref{eq:q-2 r} for $Q^{\pm}_{r,n}$. As we have seen in the proof of Theorem \ref{soluciones2 r}, these expressions are obtained by applying Darboux--Crum transformations with $\phi_{2,r,n}$ to $\phi^+_{r,n}$ and $\phi^-_{r,n}$, see expressions \eqref{eq:DT phi+} and \eqref{eq:DT phi-}. For our present discussion, we consider the unadjusted relations given in Lemma \ref{lem-Qn}:
\begin{gather}
\label{eq:q+2 sin r} Q^+_{n+1} = \dfrac{\lambda Q^+_{n} \theta_{n+1} + Q^+_{n,x} \theta_{n+1} - Q^+_{n} \theta_{n+1, x}}{\theta_{n}}, \\
\label{eq:q-2 sin r} Q^-_{n+1} = \dfrac{\lambda Q^-_{n} \theta_{n+1} - Q^-_{n,x} \theta_{n+1} + Q^-_{n} \theta_{n+1, x}}{\theta_{n}}.
\end{gather}

If we proceed in the same way performing Darboux transformations with $\phi_{1,r,n}$ we obtain that functions
\begin{gather*}
{\rm DT}(\phi_{1,r,n}){\phi}^+_{r,n} = \phi^+_{r,n,x} - \dfrac{\phi_{1,r,n,x}}{\phi_{1,r,n}} \phi^+_{r,n} \\
\hphantom{{\rm DT}(\phi_{1,r,n}){\phi}^+_{r,n}}{} = \dfrac{e^{\lambda x+(-1)^r \lambda^{2r+1} t_r}}{\theta_{r,n-1}} \dfrac{\lambda Q^+_{r,n} \theta_{r,n-1} + Q^+_{r,n,x} \theta_{r,n-1} - \theta_{r,n-1, x} Q^+_{r,n}}{\theta_{r,n}}, \\
{\rm DT}(\phi_{1,r,n}){\phi}^-_{r,n} = \phi^-_{r,n,x} - \dfrac{\phi_{1,r,n,x}}{\phi_{1,r,n}} \phi^-_{r,n} \\
\hphantom{{\rm DT}(\phi_{1,r,n}){\phi}^-_{r,n}}{} = \dfrac{e^{-\lambda x - (-1)^r \lambda^{2r+1} t_r}}{\theta_{r,n-1}} \dfrac{-\lambda Q^-_{r,n} \theta_{r,n-1} + Q^-_{r,n,x} \theta_{r,n-1} - \theta_{r,n-1, x} Q^-_{r,n}}{\theta_{r,n}},
\end{gather*}
are solutions of Schr\"odinger equation for $E \neq 0$ and potential
\begin{gather}\label{eq:crum un-1r}
 {\rm DT}(\phi_{1,r,n}){u_{r,n}} = u_{r,n} - 2 (\log \phi_{1,r,n})_{xx} = u_{r,n-1}.
\end{gather}
In the same way that we did for the functions \eqref{eq:q+2 r} and \eqref{eq:q-2 r}, we can prove that the expressions
\begin{gather*}
Q^+_{r,n-1} := \dfrac{\lambda Q^+_{r,n} \theta_{r,n-1} + Q^+_{r,n,x} \theta_{r,n-1} - \theta_{r,n-1, x} Q^+_{r,n}}{\lambda^2 \theta_{r,n}}, \\
Q^-_{r,n-1} := \dfrac{\lambda Q^-_{r,n} \theta_{r,n-1}- Q^-_{r,n,x} \theta_{r,n-1} + \theta_{r,n-1, x} Q^-_{r,n}}{\lambda^2 \theta_{r,n}}
\end{gather*}
satisfy differential systems (\ref{eq:q+1,x r})--(\ref{eq:q+1,t r}) and (\ref{eq:q-1,x r})--(\ref{eq:q-1,t r}), respectively, for $n-1$. So, we obtain
\begin{gather*}
{\rm DT}(\phi_{1,r,n}){\phi}^+_{r,n} = \phi^+_{r,n,x} - \dfrac{\phi_{1,r,n,x}}{\phi_{1,r,n}} \phi^+_{r,n} = \lambda^2 \phi^+_{r,n-1}, \\
{\rm DT}(\phi_{1,r,n}){\phi}^-_{r,n} = \phi^-_{r,n,x} - \dfrac{\phi_{1,r,n,x}}{\phi_{1,r,n}} \phi^-_{r,n} =- \lambda^2 \phi^-_{r,n-1}.
\end{gather*}
For our present discussion, we just write
\begin{gather}
\label{eq:q+3 sin r} Q^+_{n-1} = \dfrac{\lambda Q^+_{n} \theta_{n-1} + Q^+_{n,x} \theta_{n-1} - \theta_{n-1, x} Q^+_{n}}{\lambda^2 \theta_{n}}, \\
\label{eq:q-3 sin r} Q^-_{n-1} = \dfrac{\lambda Q^-_{n} \theta_{n-1}- Q^-_{n,x} \theta_{n-1} + \theta_{n-1, x} Q^-_{n}}{\lambda^2 \theta_{n}}.
\end{gather}

Now, we can prove the following result:

\begin{Theorem}\label{thm-recursionQn}
Functions $Q^+_{n} (x, t_r, \lambda)$ and $Q^-_{n} (x, t_r, \lambda)$ satisfy the differential recursions
\begin{gather}
\label{eq:recQ+}
Q^+_{0} = 1, \qquad Q^+_{1} =\lambda x - 1, \qquad Q^+_{n+1,x} Q^+_{n-1} - Q^+_{n+1} Q^+_{n-1,x} = (2n+1) Q^{+ \, 2}_{n}, \\
\label{eq:recQ-} Q^-_{0} = 1, \qquad Q^-_{1} = \lambda x + 1, \qquad Q^-_{n+1,x} Q^-_{n-1} - Q^-_{n+1} Q^-_{n-1,x} = (2n+1) Q^{- \, 2}_{n}.
\end{gather}
\end{Theorem}

\begin{proof}In Remark \ref{rem:tabla E no 0 taus} we have computed $\phi^+_{n}$ and $\phi^-_{n}$ for $n=0$ and $1$. We have obtained $Q^\pm_{0} = 1$, $Q^+_{1} = \lambda x -1$ and $Q^-_{1} = \lambda x +1$.
So, we just have to prove the recursion formulas. First, we prove \eqref{eq:recQ+}. For this, we compute $Q^+_{n+1,x}$ and $Q^+_{n-1,x}$ using expressions \eqref{eq:q+2 sin r} and \eqref{eq:q+3 sin r}:
\begin{gather*}
Q^+_{n+1,x} = \frac{1}{\theta^2_{n}}\big(( \lambda Q^+_{n,x} \theta_{n+1} + \lambda Q^+_{n} \theta_{n+1,x} + Q^+_{n,xx} \theta_{n+1} -Q^+_{n} \theta_{n+1,xx} ) \theta_{n} \\
\hphantom{Q^+_{n+1,x} =}{} + ( Q^+_{n} \theta_{n+1,x} - \lambda Q^+_{n} \theta_{n+1} - Q^+_{n,x} \theta_{n+1}) \theta_{n,x}\big),\\
Q^+_{n-1,x} = \frac{1}{\lambda^2 \theta^2_{n}} \big((\lambda Q^+_{n,x} \theta_{n-1} + \lambda Q^+_{n} \theta_{n-1,x} + Q^+_{n,xx} \theta_{n-1} - Q^+_{n} \theta_{n-1,xx}) \theta_{n} \\
\hphantom{Q^+_{n-1,x} =}{} + (Q^+_{n} \theta_{n-1,x} - \lambda Q^+_{n} \theta_{n-1} - Q^+_{n,x} \theta_{n-1}) \theta_{n,x}\big).
\end{gather*}
Replacing this expressions in the recursion formula \eqref{eq:recQ+} we get
\begin{gather*}
 Q^+_{n+1,x} Q^+_{n-1} - Q^+_{n+1} Q^+_{n-1,x} = \dfrac{\big(\lambda^2 Q^+_{n} + 2\lambda Q^+_{n} Q^+_{n,x} + Q^+_{n} Q^+_{n,xx}\big)(\theta_{n+1,x} \theta_{n-1} - \theta_{n+1} \theta_{n-1,x})}{\lambda^2 \theta_{n}^3} \\
\hphantom{Q^+_{n+1,x} Q^+_{n-1} - Q^+_{n+1} Q^+_{n-1,x} =}{} + \dfrac{\big( \lambda Q^{+\, 2}_{n} + Q^+_{n} Q^+_{n,x}\big)( \theta_{n+1} \theta_{n-1,xx} - \theta_{n+1,xx} \theta_{n-1} )}{\lambda^2 \theta_{n}^3} \\
\hphantom{Q^+_{n+1,x} Q^+_{n-1} - Q^+_{n+1} Q^+_{n-1,x} =}{} + \dfrac{ Q^{+\, 2}_{n} ( \theta_{n+1,xx} \theta_{n-1,x} - \theta_{n+1,x} \theta_{n-1,xx})}{\lambda^2 \theta_{n}^2}.
\end{gather*}
We want to compute the expressions for $\theta_{n+1}$ and $\theta_{n-1}$ in brackets in terms of $\theta_n$. The first expression is just the relation \eqref{eq:rec_dif}. Now, if we derivate with respect to $x$ expression \eqref{eq:rec_dif}, we find the second one
\begin{gather}\label{eq:rec_dif derivada}
\theta_{n+1,xx} \theta_{n-1} - \theta_{n+1} \theta_{n-1,xx} = 2 (2n+1) \theta_{n} \theta_{n,x}.
\end{gather}
In order to compute
\begin{gather}\label{eq:parentesis 3}
\theta_{n+1,xx} \theta_{n-1,x} - \theta_{n+1,x} \theta_{n-1,xx}
\end{gather}
we use relation \eqref{eq:rec2}. We have
\[ \theta_{n+1,xx} =2 \dfrac{\theta_{n+1,x} \theta_{n,x}}{\theta_{n}} - \dfrac{\theta_{n+1} \theta_{n,xx}}{\theta_{n}} \qquad \textrm{and} \qquad
\theta_{n-1,xx} = 2 \dfrac{\theta_{n-1,x} \theta_{n,x}}{\theta_{n}} - \dfrac{\theta_{n-1} \theta_{n,xx}}{\theta_{n}}. \]
Replacing both expressions in \eqref{eq:parentesis 3} we get the third one
\begin{gather}\label{eq:parentesis 3 ya}
\theta_{n+1,xx} \theta_{n-1,x} - \theta_{n+1,x} \theta_{n-1,xx} = \dfrac{\theta_{n,xx}}{\theta_{n}} (\theta_{n+1,x} \theta_{n-1} - \theta_{n+1} \theta_{n-1,x}) = (2n+1)\theta_{n} \theta_{n,xx}.
\end{gather}
Applying the expressions \eqref{eq:rec_dif}, \eqref{eq:rec_dif derivada} and \eqref{eq:parentesis 3 ya} we get
\begin{gather*}
 Q^+_{n+1,x} Q^+_{n-1} - Q^+_{n+1} Q^+_{n-1,x} = \\
 = (2n+1) \dfrac{ \big(\lambda^2 Q^{+\, 2}_{n} +2 \lambda Q^+_{n} Q^+_{n,x} + Q^+_{n} Q^+_{n,xx}\big) \theta_{n} -2 \lambda Q^{+\, 2}_{n} \theta_{n,x} - 2Q^+_{n} Q^+_{n,x} \theta_{n,x} +Q^{+ \, 2}_{n} \theta_{n,xx}}{\lambda^2 \theta_{n}}.
\end{gather*}
Finally, the expression \eqref{eq:q+1,x r} for $Q^+_{n,xx}$ yields to
\[ Q^+_{n+1,x} Q^+_{n-1} - Q^+_{n+1} Q^+_{n-1,x} = (2n+1)Q^{+\, 2}_{n}.\]
Analogously, the second recursion formula can be proved. So we have established our result.
\end{proof}

\begin{Remark}By Lemmas \ref{lem-Qn} and \ref{polinomios} for $F= \nC(\lambda, t_r)$ and $a= \lambda, \, b=-1$, we can conclude from this theorem that the functions $Q^\pm_{n} (x, t_r, \lambda)$ are polynomials of~$x$ and $\lambda$ with coefficients in $\nC( t_r)$ for all $n$. Indeed, their degree as polynomials of $\lambda$ is $n$.
Thus, Theorems \ref{soluciones2 r} and \ref{thm-recursionQn} determine the algebraic structure of $\phi^+_{r,n}$ and $\phi^-_{r,n}$.

Since polynomials $Q^{\pm}_n$ are not adjusted to any level of the KdV hierarchy, when we iterate the recurrences~\eqref{eq:recQ+} and~\eqref{eq:recQ-} we will obtain integration constants of $x$ which may depend on~$\lambda$ and $\tau_2, \ldots, \tau_n$. We will denote such integration constants by $\tau^{\pm}_2, \ldots, \tau^{\pm}_n$.
\end{Remark}

\begin{Example}For the first polynomials we find
\[ \begin{matrix} n & \qquad Q^+_{n} & \qquad Q^-_{n} \\[8pt]
 0 & \qquad 1 & \qquad 1 \\[3pt]
 1 & \qquad \lambda x -1 & \qquad \lambda x +1 \\[5pt]
 2 & \qquad \lambda^2 x^3 -3\lambda x^2 +3x + \tau^+_2 & \qquad \lambda^2 x^3 +3\lambda x^2 +3x + \tau^-_2 \\[5pt]
 3 & \qquad Q^+_{3} & \qquad Q^-_{3} \end{matrix} \]
where
\begin{gather}
 Q^+_3 = \lambda^3 x^6 - 6\lambda^2 x^5 +15\lambda x^4 -15x^3 + 5\lambda x^3 \tau^+_2 - 15 x^2 \tau^+_2 - \big(\lambda \tau^+_3 +5(\tau^+_2)^2 \big) x +\tau^+_3,\nonumber\\
 Q^-_3 = \lambda^3 x^6 + 6\lambda^2 x^5 +15\lambda x^4 +15x^3 + 5\lambda x^3 \tau^-_2 + 15 x^2 \tau^-_2 + \big(\lambda \tau^-_3 +5(\tau^-_2)^2 \big) x +\tau^-_3.\label{eq-Q3+-}
\end{gather}
\end{Example}

\subsection[Examples of fundamental matrices for the case $E \neq 0$]{Examples of fundamental matrices for the case $\boldsymbol{E \neq 0}$}

We can compute fundamental matrices for system~\eqref{eq:systemkdv1E r} for any $n$ using recursion formulas~\eqref{eq:recQ+} and~\eqref{eq:recQ-}.

\begin{Example}\label{ex-sol-E} We present explicit computations using SAGE for the fundamental solutions of the system~\eqref{eq:systemkdv1E r} when $E = -\lambda^2 \neq 0$ for same potentials as in Example \ref{ex:E=0}.

1. We first expose examples of unadjusted fundamental solutions:{\samepage
\begin{gather*} \begin{matrix}n & \quad \phi^+_{r,n} & \phi^-_{r,n} \\[5pt]
 0 & e^{\lambda x +(-1)^r \lambda^{2r+1} t_r} & e^{-\lambda x -(-1)^r \lambda^{2r+1} t_r} \\[3pt]
 1 & e^{\lambda x +(-1)^r \lambda^{2r+1} t_r} \dfrac{\lambda x -1}{x} & e^{-\lambda x -(-1)^r \lambda^{2r+1} t_r} \dfrac{\lambda x +1}{x} \\[13pt]
 2 & e^{\lambda x +(-1)^r \lambda^{2r+1} t_r} \dfrac{\lambda^2 x^3 -3\lambda x^2 +3x + \tau^+_2}{x^3 + \tau_2} & e^{-\lambda x -(-1)^r \lambda^{2r+1} t_r} \dfrac{\lambda^2 x^3 +3\lambda x^2 +3x + \tau^-_2}{x^3 + \tau_2}
\\[13pt] 3 & e^{\lambda x +(-1)^r \lambda^{2r+1} t_r} \dfrac{Q^+_3 (\lambda,x, t_r) }{x^6 + 5x^3 \tau_2 + x \tau_3 -5 \tau_2^2} & e^{-\lambda x -(-1)^r \lambda^{2r+1} t_r}\dfrac{Q^-_3 (\lambda,x, t_r)}{x^6 + 5x^3 \tau_2 + x \tau_3 -5 \tau_2^2}
\end{matrix} \end{gather*}
where $Q^+_3$ and $Q^-_3 $ are the ones given in \eqref{eq-Q3+-}.}

2. Next, we expose fundamental solutions for potentials which are solutions of the first level of the KdV hierarchy, KdV$_1$ equation: $ u_{t_1} = \frac{3}{2} uu_x - \frac{1}{4} u_{xxx}$. We also show the explicit choice of the functions $\tau^{\pm}_i$. The choice of functions $\tau_i$ is the same as in Example~\ref{ex:E=0}:
\begin{gather*} \begin{matrix}n & \phi^+_{1,n} & \phi^-_{1,n} & \!\!\!\!\big(\tau^{\pm}_2, \ldots,\tau^{\pm}_{n}\big)\! \\[5pt]
 0\! & e^{\lambda x -\lambda^{3} t_1} & e^{-\lambda x + \lambda^{3} t_1} & \\[3pt]
 1\! & e^{\lambda x - \lambda^{3} t_1} \dfrac{\lambda x -1}{x} & e^{-\lambda x +\lambda^3 t_1} \dfrac{\lambda x +1}{x} & \\[13pt]
 2\! & \!\!e^{\lambda x {-} \lambda^{3} t_1} \dfrac{\lambda^2 x^3 {-}3\lambda x^2 {+}3x {+}3\lambda^2 t_1}{x^3 +3t_1}\!\! & \!\!e^{-\lambda x {+}\lambda^3 t_1} \dfrac{\lambda^2 x^3 {+}3\lambda x^2 {+}3x {+}3\lambda^2 t_1}{x^3 +3t_1}\!\! & \!\!\!\!\big(3 \lambda^2 t_1\big)\!
\\[13pt]
3\! & e^{\lambda x - \lambda^{3} t_1} \dfrac{Q^+_3 (\lambda,x, t_1) }{x^6 + 15x^3t_1 -45t_1^2} & e^{-\lambda x +\lambda^3 t_1}\dfrac{Q^-_3 (\lambda,x, t_1)}{x^6 + 15x^3t_1 -45t_1^2} &
\!\!\!\!\big(3 \lambda^2 t_1,\!{-}45 \big(\lambda^3t^2_1 {\pm} t_1\big) \big)\!
\end{matrix} \end{gather*}
where
\begin{gather*}
Q^+_3 (\lambda,x, t_1) = \lambda^3 x^6 - 6\lambda^2 x^5 +15\lambda x^4 -15x^3 + 15\lambda^3 x^3 t_1 - 45\lambda^2 x^2 t_1 + 45\lambda xt_1 \\
\hphantom{Q^+_3 (\lambda,x, t_1) =}{} - 45\lambda^3 t_1^2 - 45t_1,\\
Q^-_3 (\lambda,x, t_1) = \lambda^3 x^6 + 6\lambda^2 x^5 +15\lambda x^4 +15x^3 + 15\lambda^3 x^3 t_1 +45\lambda^2 x^2 t_1 + 45\lambda xt_1 \\
\hphantom{Q^-_3 (\lambda,x, t_1) =}{} - 45\lambda^3 t_1^2 +45t_1.
\end{gather*}
\end{Example}

\section{Spectral curves and Darboux--Crum transformations}\label{sec espec darb}

Let $\Gamma_n \subset \nC^2$ be the spectral curve associated to the stationary Schr\"odinger operator $-\partial_{xx} +u -E$ where $u $ is a s-KdV$_n$ potential. Next we consider the Zariski closure of $\Gamma_n$, say $\overline{\Gamma}_{n}$, in the complex projective plane $ \mathbb{P}^2 $. Let be
\begin{gather*}
p (E, \mu)=\mu^2-R_{2n+1} (E)= \mu^2-\sum_{j=0}^{2n+1} C_{j} E^j =0
\end{gather*} an equation for~$\Gamma_n$. Then an equation for $\overline{\Gamma}_{n}$ is
\begin{gather*}
 p_h (E,\mu,\nu ) = \mu^2 \nu^{2n-1} - \widehat{R}_{2n+1} (E, \nu) =0,
\end{gather*}
where
\begin{gather*}
\widehat{R}_{2n+1} (E, \nu) = \nu^{2n+1} R_{2n+1} \left ( \frac{E}{\nu} \right) = \sum_{j=0}^{2n+1} C_{j} \nu^{2n+1-j} E^j
\end{gather*} is an homogeneous polynomial of degree $2n+1$. Moreover, observe that the singular points of~$\overline{\Gamma}_{n}$ are
\begin{gather*}
 \Sing \big( \overline{\Gamma}_{n} \big)=
 \{(E,0) \colon E \ \text{is a multiple root of}\ R_{2n+1}
 \} \cup \{\ P_\infty =[0:1:0] \},
\end{gather*}
and also
\begin{gather}\label{eq:infty E0}
 \overline{\Gamma}_n \cap \{ E=0 \} = \big\{ [0: \mu : \nu] \in \mathbb{P}^2\colon \mu^2 \nu^{2n-1}= C_0 \nu^{2n+1} \big\}.
\end{gather}

\subsection{Extended Green's function} \label{ssec:ext green}

Following \cite{GH}, we define the Green's function on $\Gamma_n \times \nC$ as
\begin{gather*}
g(E,\mu,x )= \dfrac{\phi_1 \phi_2}{W(\phi_1, \phi_2)},
\end{gather*}
where $\phi_1$ and $\phi_2$ are two independent solutions of Schr\"odinger equation
\begin{gather}\label{eq:schr0 spec}
(L-E) \phi =(-\partial_{xx} +u -E) \phi = 0.
\end{gather}
for the same value of $E$ and $W(\phi_1, \phi_2) $ stands for their Wronskian.

Let
\begin{gather}\label{eq:sigma conj}
 \sigma_+ = \sigma (E, \mu )= \dfrac{i\mu + F_{n,x}/2}{ F_n}, \qquad \sigma_{-} = \sigma (E, -\mu )= \dfrac{-i\mu + F_{n,x}/2}{F_n}
\end{gather}
be functions defined over the spectral curve. We recall the following result:

\begin{Lemma}[{\cite[Lemma 1.8]{GH}}]\label{lm:chapter3 GH}
Let $u$ be solution of s-KdV$_n$ equation \eqref{eq:skdvn2}. Let $\phi_1$ and $\phi_2$ be solutions of Schr\"odinger equation \eqref{eq:schr0 spec} for this potential and with corresponding functions over the spectral curve $\sigma_+$ and $\sigma_-$ defined by~\eqref{eq:sigma conj}. Then $\sigma_+$ and $\sigma_-$ are solutions of the Riccati type equation
\begin{gather}\label{eq:riccati GH}
 \sigma^2 + \sigma_x = u -E.
\end{gather}
Moreover, the following identities are satisfied
\begin{gather}
\sigma_+ + \sigma_- = \dfrac{F_{n,x}}{F_n} = \dfrac{(\phi_1 \phi_2)_x}{\phi_1 \phi_2}, \qquad
\sigma_+ - \sigma_- = \dfrac{2i\mu}{F_n} = - \dfrac{W(\phi_1, \phi_2)}{\phi_1 \phi_2},\nonumber \\
\sigma_+ \cdot \sigma_- = \dfrac{H_n}{F_n} = \dfrac{\phi_{1,x} \phi_{2,x}}{\phi_1 \phi_2},\label{eq:GH 1}
\end{gather}
where $W(\phi_1, \phi_2) = \phi_1 \phi_{2,x} - \phi_{1,x} \phi_{2}$ denotes the Wronskian of $\phi_1$ and $\phi_2$.
\end{Lemma}

We remark that this lemma is essentially a reformulation of a classic result that goes back to Hermite when he was studying closed form solutions for Lam\'e equation~\cite{HE}. In~\cite{WW} call this approach the Lindeman--Stieljes theory but, as far as we know, this approach was used for the first time by Hermite, and then by others: Halphen, Brioschi, Crawford, Stieljes, \dots. The method used that the product of solutions $X=\phi_1 \phi_2$ is a solution of the second symmetric power of the Schr\"odinger equation
\begin{gather} \label{symmp} (-\partial_{xxx}-4(u-E)\partial_x-2u_x)X=0.\end{gather}
Then the relations~\eqref{eq:GH 1} connect the solutions of the Riccati equation with that of the second symmetric power. The fact that there is a connection between the solutions of the second symmetric product and the Riccati equation of the Schr\"odinger equation is relevant for the differential Galois theory, although we will not use explicitely this connection in this paper. Furthermore it is interesting to point out that the solutions of the Lam\'e equation obtained by Hermite in~\cite{HE}, are associated to other algebro-geometric solutions of KdV, finite-gap solutions with regular spectral curves, see~\cite{MoRuZu} and references therein.
As far as we know, the relevance of the equation~\eqref{symmp} for the KdV equation was considered for the first time by Gel'fand and Dikii in their fundamental paper about the asymptotic behaviour of the resolvent of the Schr\"odinger equation associated to the KdV equation~\cite{GD}.

By Lemma \ref{lm:chapter3 GH}, the Green's function can be rewritten as
\begin{gather}\label{eq-Green1}
g(E,\mu,x )= \dfrac{i F_n (E,x)}{2 \mu} = \dfrac{1}{\sigma_- - \sigma_+}.
\end{gather}
Observe that $g$ is well defined whenever $\mu\not=0$, i.e., for energy levels such that $R_{2n+1} (E) \not=0$.

Next, let define a extension of $g$ on $\overline{\Gamma}_{n} \times \nC_x $ as
\begin{gather*}
g_h (E,\mu,\nu, x )= \dfrac{i \nu^n F_n (E/ \nu,x)}{2 \mu \nu^{n-1}}, \qquad \text{for} \quad [E:\mu :\nu ] \in \overline{\Gamma}_{n} \setminus \{ \mu\nu=0 \}.
\end{gather*}
We call $g_h$ {\it the homogenized Green's function}. Next we will show that $g_h $ is well defined and also that it extends $g$, that is $g_h (E,\mu,1,x)=g(E,\mu,x )$ for $(E,\mu,x)\in \Gamma_n \times \nC_x $. To do that, observe that
\[ g_h (E, \mu, 1, x) = g(E,\mu,x ) \qquad \text{and} \qquad g_h (a E,a \mu,a \nu, x )=g_h (E,\mu,\nu, x ), \]
for any $a\in\nC$, $a\not= 0$. Moreover, we have that
\begin{gather}\label{eq:chapter3 Fn homog}
 \widehat{F}_n (E, \nu,x): = \nu^n F_n (E/ \nu,x) = \sum_{j=0}^n f_{n-j} \nu^{n-j} E^j
\end{gather}
is an homogeneous polynomial in $E$ of degree $n$ and then
\begin{gather*}
g_h (E,\mu,\nu, x )= \dfrac{i \widehat{F}_n (E, \nu,x) }{2 \mu \nu^{n-1}}, \qquad \text{for} \quad [E:\mu:\nu ]\in \overline{\Gamma}_{n}.
\end{gather*}
Also, we get the following formula
\begin{gather}\label{eq:curva hatF}
 \mu^2 \nu^{2n-2} = \nu^{2n} R_{2n+1} (E/\nu) = \dfrac{\nu \widehat{F}_n\widehat{F}_{n,xx}}{2} - (u - E/\nu) \widehat{F}_n^2 - \dfrac{\nu^2 \widehat{F}^2_{n,x}}{4},
\end{gather}
where
\begin{gather}\label{eq:chapter3 Fnx homog}
 \widehat{F}_{n,x} = \nu^{n-1} F_{n,x} (E/\nu) \qquad \text{and} \qquad \widehat{F}_{n,xx} = \nu^{n-1} F_{n,xx} (E/\nu)
\end{gather}
are homogeneous polynomials in $E$ and $\nu$ of degree $n-1$.

Now, consider equation \eqref{eq:kdv est}
\[ 0= \dfrac{F_{n,xxx}}{2} - 2(u-E) F_{n,x} -u_x F_n, \]
after multiplication by $F_n$ and integration, this equation reads
\begin{gather*}
c = \dfrac{F_n F_{n,xx}}{2} -(u-E) F^2_n - \dfrac{F^2_{n,x}}{4},
\end{gather*}
where $c$ is a integration constant. By~\eqref{eq-Green1} we have the following differential relation for the function~$g$:
\begin{gather*}
\dfrac{1}{2}g g_{xx} -(u-E) g^2 - \dfrac{1}{4} g^2_{x} = - \dfrac{1}{4},
\end{gather*}
since $g_x = (\sigma_+ + \sigma_-)g$ and $g_{xx} = 2 (u-E + \sigma_+ \sigma_- ) g$.

Now let define the extensions of $\sigma_+$ and $\sigma_-$ on $\overline{\Gamma}_{n} \times \nC_x $ as
\begin{gather}\label{eq:sigmah conj}
 (\sigma_+ )_h= \dfrac{i\mu \nu^{n-1} +\nu \widehat{F}_{n,x}/2}{ \widehat{F}_n}, \qquad
 (\sigma_{-} )_h= \dfrac{-i\mu \nu^{n-1} +\nu \widehat{F}_{n,x}/2}{\widehat{F}_n},
\end{gather}
where we have used previous notation. Notice that the functions $(\sigma_+ )_h$ and $(\sigma_- )_h$ are solutions of the Riccati type equation
\[ ((\sigma_{\pm})_h)^2 + ((\sigma_{\pm})_x )_h = u -E/\nu. \]
Moreover we have that the function
\begin{gather*}
g_h= \dfrac{i \widehat{F}_n (E, \nu,x) }{2 \mu \nu^{n-1}} =\dfrac{1}{(\sigma_- )_h - (\sigma_+ )_h }
\end{gather*}
is a solution of
\begin{gather*}
\dfrac{1}{2} g_h (g_{xx})_h -(u-E/ \nu) g_h^2 - \dfrac{1}{4} \big(g^2_{x}\big)_h = - \dfrac{1}{4}.
\end{gather*}

\subsubsection{Transformed Green's functions}

Now, we analyze how Darboux--Crum transformations change Green's functions $g$ and $g_h$. For that, we will use solutions of the Riccati type equation \eqref{eq:riccati GH} as a esential tool.

Let $u$ be solution of s-KdV$_n$ equation \eqref{eq:skdvn2}. Let $\phi_1$ and $\phi_2$ be solutions of Schr\"odinger equation~\eqref{eq:schr0 spec} for this potential and energy level $E$. Next we consider $\phi_0$ a solution of Schr\"odinger equation for $u$ and $E_0$, with $E_0 \neq E$ and choose as corresponding point of the spectral curve $(E_0, \mu_0)$. Recall that after applying a Darboux--Crum transformation with $\phi_0$ to $u$, $\phi_1$ and $\phi_2$, we get
\begin{gather*}
{\rm DT}(\phi_0) u = u- 2 \sigma_{0,x}, \qquad
{\rm DT}(\phi_0) \phi_1 = \phi_{1,x} - \sigma_0 \phi_1, \qquad
{\rm DT}(\phi_0) \phi_2 = \phi_{2,x} - \sigma_0 \phi_2,
\end{gather*}
where $\sigma_0 = (\log \phi_0)_x$ is a solution of the Riccati equation $ \sigma^2 + \sigma_x = u -E_0$. By Lemma~\ref{lm:chapter3 GH}, the function $\sigma^0$ equals
\begin{gather}\label{eq:chapter3 sigma0}
\sigma^0 = \sigma (E_0, \mu_0) = \dfrac{i\mu_0 + F^0_{n,x}/2}{F^0_n},
\end{gather}
 where $F_n^0 = F_n (E_0)$, is a solution of the same Riccati equation for $E=E_0$. Thus, we conclude that we can perform a Darboux transformation using $\sigma^0$ instead of $\sigma_0$. The transformed functions
\begin{gather*}
\widetilde{\phi}_1 = \phi_{1,x} - \sigma^0 \phi_1 \qquad \text{and} \qquad \widetilde{\phi}_2 = \phi_{2,x} - \sigma^0 \phi_2
\end{gather*}
are solutions of the Schr\"odinger equation for potential
\[ \widetilde{u} = u- 2 \sigma^0_x. \]

Now, we take the functions $\sigma_1 = ( \log \phi_1)_x$ and $\sigma_2 = ( \log \phi_2)_x$, which are solutions of the Riccati equation \eqref{eq:riccati GH} for $E \neq E_0$. Then, by equations \eqref{eq:GH 1}, we get the equalities
\begin{gather}
\label{eq:chapter3 sigma+-1} \sigma_+ - \sigma_- = \dfrac{2i\mu}{F_n} =- \dfrac{W(\phi_1, \phi_2)}{\phi_1 \phi_2} = \dfrac{\phi_{1,x}}{\phi_1} - \dfrac{\phi_{2,x}}{\phi_2} = \sigma_1 - \sigma_2, \\
\label{eq:chapter3 sigma+-2} \sigma_+ + \sigma_- = \dfrac{F_{n,x}}{F_n} = \dfrac{\phi_1 \phi_{2,x} + \phi_{1,x} \phi_2}{\phi_1 \phi_2} = \dfrac{\phi_{1,x}}{\phi_1} + \dfrac{\phi_{2,x}}{\phi_2} = \sigma_1 + \sigma_2, \\
\label{eq:chapter3 sigma+-3} \sigma_+ \cdot \sigma_- = \dfrac{\phi_{1,x} \phi_{2,x}}{\phi_1 \phi_2} = \dfrac{\phi_{1,x}}{\phi_1} \dfrac{\phi_{2,x}}{\phi_2} = \sigma_1 \sigma_2.
 \end{gather}

Next we define the transformed Green's function
\begin{gather*}
\widetilde{g}(E, \mu, x) = \dfrac{\widetilde{\phi}_1 \widetilde{\phi}_2}{W(\widetilde{\phi}_1,\widetilde{\phi}_2)}.
\end{gather*}
The relations \eqref{eq:chapter3 sigma+-1}--\eqref{eq:chapter3 sigma+-3} link the Green's functions as follows
\[ \widetilde{g}(E, \mu, x) = \dfrac{ \big(\sigma_1 - \sigma^0\big)\big(\sigma_2 - \sigma^0\big)}{(E- E_0) } \cdot \dfrac{ \phi_1 \phi_2}{W(\phi_1, \phi_2)} = \dfrac{ \big(\sigma_+ - \sigma^0\big) \big(\sigma_- - \sigma^0\big)}{(E- E_0)} g(E,\mu,x). \]
Hence we obtain a rational presentation of $\widetilde{g}$ as a consequence of the formulas~\eqref{eq:chapter3 sigma0} and~\eqref{eq:sigma conj}. We write this formula in \eqref{eq:chapter3 green funct TD}.

\begin{Proposition}\label{prop-Green1}
The Green's function associated to the transformed Schr\"odinger operator explicitly reads
\begin{gather}\label{eq:chapter3 green funct TD}
\widetilde{g}(E, \mu, x) = \dfrac{i \left (\mu^2 \big(F_n^0\big)^2 -\mu_0^2 F_n^2 -i\mu_0 F_n \big(F_n^0 F_{n,x} - F^0_{n,x} F_n\big) + \frac{(F_n^0 F_{n,x} - F^0_{n,x} F_n)^2}{4} \right ) }{2 \mu (E-E_0) F_n \big(F_n^0\big)^2}.
\end{gather}
\end{Proposition}

\begin{Remark}
Observe that for $E_0=0$ the formula \eqref{eq:chapter3 green funct TD} becomes
\begin{gather*}
\widetilde{g}(E, \mu, x) = \dfrac{i \left (\mu^2 f_n^2 -\mu_0^2 F_n^2 -i\mu_0 F_n (f_n F_{n,x} - f_{n,x} F_n) + \frac{(f_n F_{n,x} - f_{n,x} F_n)^2}{4} \right ) }{2 \mu E F_n f_n^2}.
\end{gather*}
\end{Remark}

We will use the following result from \cite{GH}.

\begin{Proposition}[{\cite[Lemma G.1]{GH}}]\label{lemma-GH}
Let $u$ be solution of s-KdV$_n$ equation, let $(E_0, \mu_0)$ and $(E, \mu)$ be two different points of $\Gamma_n$. Then the transformed Green's function explicitly reads
\begin{gather*}
\widetilde{g}(E, \mu, x) = \dfrac{ \big(\sigma_+ - \sigma^0\big)\big(\sigma_- - \sigma^0\big)}{(E- E_0)} \dfrac{ i F_n}{2 \mu } = \dfrac{i \widetilde{F}_{\widetilde{n}} (E,x)}{2 \widetilde{\mu}},
\end{gather*}
where $\widetilde{F}_{\widetilde{n}}$ is a polynomial in $E$ of degree $\widetilde{n}$ and $\widetilde{\mu}$ is such that $\Gamma_{\widetilde{n}}\colon \widetilde{\mu}^2 - \widetilde{R}_{2\widetilde{n} +1} =0$ for some polynomial $\widetilde{R}_{2\widetilde{n} +1} (E)$ of degree $2\widetilde{n} +1$, with $0 \leq \widetilde{n} \leq n+1$.
\end{Proposition}

Next, for the homogeneized Green's function, choose the point of the spectral curve $[E_0 : \mu_0 : \nu_0]$. We define the extension of $\sigma^0$ on $\overline{\Gamma}_{n} \times \nC_x $ as
\begin{gather*}
 \big(\sigma^0 \big)_h (E_0,\mu_0, \nu_0)= \dfrac{i\mu_0 \nu_0^{n-1} +\nu_0 \widehat{F}^0_{n,x}/2}{ \widehat{F}^0_n},
\end{gather*}
where $\widehat{F}^0_n = \widehat{F}_n (E_0, \nu_0,x) $ for $\widehat{F}_n (E, \nu,x) $ defined by \eqref{eq:chapter3 Fn homog} and $\widehat{F}^0_{n,x} = \widehat{F}_{n,x} (E_0, \nu_0,x) $, for $\widehat{F}_{n,x}$ defined in \eqref{eq:chapter3 Fnx homog}. Notice that when $\nu_0 =0$ function $\big(\sigma^0 \big)_h$ vanishes. So, whenever $\nu_0 =0$ we define
\[ (\sigma^0 )_h (E_0,\mu_0, 0) := 0, \qquad \text{for} \quad [E_0 : \mu_0 :0] \in \overline{\Gamma}_n. \]

Using above notation we have the following results.

\begin{Proposition}\label{prop-Green2}
Let assume $C_0 = R_{2n+1} (0)\not=0$.
For $E_0=0$ and $\mu_0 \neq 0$, the homogeneized Green's function associated to the transformed Green's function $\widetilde{g}$ for $-\partial_{xx} +\widetilde{u} -E $ explicitly reads
\begin{gather*}
(\widetilde{g})_h (E, \mu, \nu, x) = \dfrac{i \left ( \frac{\nu^2 \widehat{F}_{n,xx}}{2} +(E- \nu u ) \widehat{F}_n + \frac{\nu f_{n,x}^2 \widehat{F}_n}{4 f_n^2} - \frac{\nu^2 f_{n,x} \widehat{F}_{n,x}}{2 f_n}
- \frac{\nu C_0 \widehat{F}_n}{f_n^2}\right ) }{2 \mu E \nu^{n-1} } \nonumber \\
\hphantom{(\widetilde{g})_h (E, \mu, \nu, x) =}{} + \dfrac{ C_0 \nu_0 ( \nu f_n \widehat{F}_{n,x}- f_{n,x} \widehat{F}_n) }{2 \mu E \nu^{n-2} \mu_0 f_n^2 },
 \end{gather*}
where $\widehat{F}_n (E, \nu,x) $ is defined by~\eqref{eq:chapter3 Fn homog} and $\widehat{F}_{n,x} (E, \nu,x)$, $\widehat{F}_{n,xx} (E, \nu,x)$ are defined by~\eqref{eq:chapter3 Fnx homog}.
\end{Proposition}

\begin{Remark}Formula
\[
\dfrac{\nu^2 \widehat{F}_{n,xx}}{2} +(E- \nu u ) \widehat{F}_n + \dfrac{\nu f_{n,x}^2 \widehat{F}_n}{4 f_n^2} - \dfrac{\nu^2 f_{n,x} \widehat{F}_{n,x}}{2 f_n}
- \dfrac{\nu C_0 \widehat{F}_n}{f_n^2}
\]
is an homogeneous polynomial in $E$ and $\nu$ of degree $n+1$.
\end{Remark}

\begin{proof} First, consider the transformed Green's function $\widetilde{g}$ given by \eqref{eq:chapter3 green funct TD}. Then, the homogenized Green's function is obtained by the homogenization process as
\begin{gather*}
(\widetilde{g})_h (E, \mu, \nu, x) = \left ( \dfrac{ \big(\sigma_+ - \sigma^0\big)\big(\sigma_- - \sigma^0\big)}{(E- E_0)} \dfrac{ i F_n}{2 \mu } \right )_h \\
\hphantom{(\widetilde{g})_h (E, \mu, \nu, x)}{} = \dfrac{i \left (\mu^2 \nu^{2n-2} \big(\widehat{F}_n^0\big)^2 + \frac{(\nu \widehat{F}_n^0 \widehat{F}_{n,x} -\nu_0 \widehat{F}^0_{n,x} \widehat{F}_n)^2}{4}\right )}{2 \mu \nu^{n-1} (E/ \nu -E_0/ \nu_0) \widehat{F}_n \big(\widehat{F}_n^0\big)^2} \\
\hphantom{(\widetilde{g})_h (E, \mu, \nu, x)=}{} - \dfrac{ i \mu_0 \nu_0^{n-1} \big( \mu_0 \nu_0^{n-1} \widehat{F}_n^2 +i \widehat{F}_n \big(\nu \widehat{F}_n^0 \widehat{F}_{n,x}- \nu_0 \widehat{F}^0_{n,x} \widehat{F}_n\big) \big)}{2 \mu \nu^{n-1} (E/ \nu -E_0/ \nu_0) \widehat{F}_n \big(\widehat{F}_n^0\big)^2},
\end{gather*}
where $\widehat{F}_n (E, \nu,x) $ is defined by~\eqref{eq:chapter3 Fn homog}, $\widehat{F}_{n,x} (E, \nu,x)$ is defined in \eqref{eq:chapter3 Fnx homog}, $\widehat{F}^0_n = \widehat{F}_n (E_0, \nu_0,x) $ and $\widehat{F}^0_{n,x} = \widehat{F}_{n,x} (E_0, \nu_0,x) $. In particular, for $E_0 =0$, we get
\begin{gather*}
 (\widetilde{g})_h (E, \mu, \nu, x) = \dfrac{i \left (\mu^2 \nu^{2n-2} f_n^2 + \frac{(\nu f_n \widehat{F}_{n,x} -f_{n,x} \widehat{F}_n)^2}{4}\right )}{2 \mu E \nu^{n-2} \widehat{F}_n f^2_n} \\
 \hphantom{(\widetilde{g})_h (E, \mu, \nu, x) =}{} - \dfrac{i \mu^2_0 \widehat{F}_n }{2 \mu E \nu^{n-2} \nu_0^{2} f^2_n} + \dfrac{ \mu_0 (\nu f_n \widehat{F}_{n,x}- f_{n,x} \widehat{F}_n) }{2 \mu E \nu^{n-2} \nu_0 f^2_n},
\end{gather*}
since $\widehat{F}_n (0, \nu_0,x) = \nu_0^n f_n$ and $\widehat{F}_{n,x} (0, \nu_0,x) = \nu_0^{n-1} f_{n,x}$. Considering~\eqref{eq:curva hatF} we get the following expression
\begin{gather*}
(\widetilde{g})_h (E, \mu, \nu, x) = \dfrac{i \left ( \frac{\nu^2 \widehat{F}_{n,xx}}{2} +(E- \nu u ) \widehat{F}_n + \frac{\nu f_{n,x}^2 \widehat{F}_n}{4 f_n^2} - \frac{\nu^2 f_{n,x} \widehat{F}_{n,x}}{2 f_n} \right ) }{2 \mu E \nu^{n-1} } \\
\hphantom{(\widetilde{g})_h (E, \mu, \nu, x) =}{} - \dfrac{ i \mu_0^2 \widehat{F}_n }{2 \mu E\nu^{n-2} \nu_0^{2} f_n^2 } +\dfrac{ \mu_0 (\nu f_n \widehat{F}_{n,x}- f_{n,x} \widehat{F}_n) }{2 \mu E \nu^{n-2} \nu_0 f_n^2 }.
\end{gather*}
Moreover, by \eqref{eq:infty E0} we have that $\mu^2_0 = C_0 \nu_0^2$, and then
\begin{gather*}
 (\widetilde{g})_h (E, \mu, \nu, x) = \dfrac{i \left ( \frac{\nu^2 \widehat{F}_{n,xx}}{2} +(E- \nu u ) \widehat{F}_n + \frac{\nu f_{n,x}^2 \widehat{F}_n}{4 f_n^2} - \frac{\nu^2 f_{n,x} \widehat{F}_{n,x}}{2 f_n} \right ) }{2 \mu E \nu^{n-1} } \\
\hphantom{(\widetilde{g})_h (E, \mu, \nu, x) =}{} - \dfrac{ i C_0 \widehat{F}_n }{2 \mu E \nu^{n-2} f_n^2 } + \dfrac{ C_0 \nu_0 (\nu f_n \widehat{F}_{n,x}- f_{n,x} \widehat{F}_n) }{2 \mu E
 \nu^{n-2} \mu_0 f_n^2 }.
\end{gather*}
And then the result follows.
\end{proof}

\begin{Proposition}\label{prop-Green3}
Let assume $C_0 = R_{2n+1} (0)=0$. For $E_0=0$ and $\mu_0 \neq 0$, the homogeneized Green's function associated to the transformed Green's function $\widetilde{g}$ for $-\partial_{xx} +\widetilde{u} -E $ explicitly reads
\begin{gather*}
(\widetilde{g})_h (E, \mu, \nu, x) = \dfrac{i\left (\frac{\nu^2 \widehat{F}_{n,xx}}{2} +(E- \nu u ) \widehat{F}_n \right )}{2 \mu E\nu^{n-1}},
\end{gather*}
where $\widehat{F}_n (E, \nu,x) $ is defined by~\eqref{eq:chapter3 Fn homog} and $\widehat{F}_{n,xx} (E, \nu,x)$ is defined in~\eqref{eq:chapter3 Fnx homog}.
\end{Proposition}

\begin{Remark}
Formula
\[
\frac{\nu^2 \widehat{F}_{n,xx}}{2} +(E- \nu u ) \widehat{F}_n
\]
is an homogeneous polynomial in $E$ and $\nu$ of degree $n+1$.
\end{Remark}

\begin{proof}When $C_0 =0$ we have that $\nu_0= 0$ by \eqref{eq:infty E0}, since $\mu_0 \neq 0$. So, $\big(\sigma^0\big)_h =0$. Hence, the homogeneized Green's function in this case is
\begin{gather*}
 (\widetilde{g})_h (E, \mu, \nu, x) = \dfrac{(\sigma_+ )_h (\sigma_-)_h}{E/\nu} \dfrac{ i \widehat{F}_n}{2 \mu\nu^{n-1} } \\
 \hphantom{(\widetilde{g})_h (E, \mu, \nu, x)}{} = \dfrac{i \left ( \frac{\mu^2 \nu^{2n-2} + \nu^2 \widehat{F}^2_{n,x}/4}{\widehat{F}_n} \right )}{2 \mu E\nu^{n-2}} = \dfrac{i\left (\frac{\nu^2 \widehat{F}_{n,xx}}{2} +(E- \nu u ) \widehat{F}_n \right )}{2 \mu E\nu^{n-1}},
\end{gather*}
by \eqref{eq:sigmah conj} and \eqref{eq:curva hatF}.
\end{proof}

\subsection{Darboux--Crum transformations for the spectral curve}\label{ssec-DT-C}

In this subsection we present how Darboux--Crum transformations affect the spectral curve $\Gamma_n $. We observe that the action of the transformation ${\rm DT}(\phi_0 )$ strongly depends on the type of point~$P$ in the spectral curve we use to construct $\phi_0$. In fact, if~$P$ is a regular point, the curve associated with the transformed potential is the same; in the other cases the new curve is a~blowing-down or a~blowing-up of~$\Gamma_n$.

\begin{Theorem}[I]\label{thm:curva y DT1} Let $(E_0, \mu_0 ) \in {\Gamma}_n $ and $u$ be a solution of s-KdV$_n$ equation. Let $\phi_0$ be a~solution of Schr\"odinger equation for energy $E_0$ and potential $u$, i.e., $\phi_{0,xx} = (u- E_0) \phi_0$. Let $\widetilde{u} = u- 2 (\log \phi_0)_{xx}$ be the Darboux--Crum transformation of $u$. Then, $\widetilde{u}$ is a solution of s-KdV$_{\widetilde{n}}$ equation for
\[
\widetilde{n} = \begin{cases}
n & \textrm{if } (E_0, \mu_0) \textrm{ is a regular point of } \Gamma_n,\\
n-1 & \textrm{if } (E_0, \mu_0) \textrm{ is an affine singular point of } \Gamma_n.
\end{cases}
\]
Furthermore, the spectral curve associated to $\widetilde{u}$ is $\Gamma_{\widetilde{n}} : \widetilde{\mu}^2 - \widetilde{R}_{2\widetilde{n} +1} =0$, with
\[
\widetilde{R}_{2\widetilde{n} +1} =
 \begin{cases}
R_{2n+1} & \textrm{if } (E_0, \mu_0) \textrm{ is a regular point of } \Gamma_n,\\
(E - E_0)^{-2}R_{2n+1} & \textrm{if } (E_0, \mu_0) \textrm{ is an affine singular point of } \Gamma_n. \end{cases} \]
\end{Theorem}

The idea of the proof is to compute the Green's function \eqref{eq:chapter3 green funct TD} associated to $\widetilde{u}$ and interpret the result by means of Lemma~\ref{lemma-GH}.

\begin{proof} First, we suppose that $(E_0, \mu_0)$ is a regular point and $\mu_0 \neq 0$. In this case, we compute
\[ \big(\sigma_+ - \sigma^0\big)\big(\sigma_- - \sigma^0\big) = \dfrac{\mu^2 \big(F_n^0\big)^2 -\mu_0^2 F_n^2 -i\mu_0 F_n \big(F_n^0 F_{n,x} - F^0_{n,x} F_n\big) + \frac{(F_n^0 F_{n,x} - F^0_{n,x} F_n)^2}{4}}{F_n^2 \big(F_n^0\big)^2}.
\]
We use Corollaries \ref{cor-G1} and \ref{cor-G2} to rewrite the expressions $F_n^0 F_{n,x} - F^0_{n,x} F_n $ and $\mu^2 \big(F_n^0\big)^2 -\mu_0^2 F_n^2$.
This yields to the equality
\[ \big(\sigma_+ - \sigma^0\big)\big(\sigma_- - \sigma^0\big) = (E - E_0) \dfrac{\frac{P_{n,x}}{2} + F_n F^0_{n} - P_n \sigma^0}{F_n F_n^0}.\]
Finally, we replace this expression in the Green's function \eqref{eq:chapter3 green funct TD}:
\[ \widetilde{g}(E, \mu, x) = \dfrac{iF_n \big(\sigma_+ - \sigma^0\big)\big(\sigma_- - \sigma^0\big)}{2 \mu (E - E_0)} = \dfrac{i \left (F_n + \frac{P_{n,x}}{2 F^0_n} - \frac{P_n \sigma^0}{F^0_n} \right )}{2 \mu} = \dfrac{i \widetilde{F}_{\widetilde{n}}}{2 \mu}. \]
Since $\widetilde{F}_{\widetilde{n}} = F_n + \frac{P_{n,x}}{2 F^0_n} - \frac{P_n \sigma^0}{F^0_n}$ is a polynomial in $E$ of degree $n$, by means of Lemma \ref{lemma-GH}, we conclude that $ \widetilde{n} =n$ and $\widetilde{\mu} = \mu$. Thus, $\widetilde{R}_{2\widetilde{n}+1} = R_{2n+1}$.

Now, we suppose that $(E_0, \mu_0)$ is a regular point and $\mu_0 = 0$. In this case, we have that $R^0_{2n+1} = R_{2n+1} (E_0) =0$ and $R^0_{2n+1, E} = \partial_E ( R_{2n+1} )(E_0) \neq 0$, thus,
\[
\mu^2 = R_{2n+1} (E) = (E-E_0) M_{2n},
\]
where $M_{2n} (E)$ is a polynomial in $E$ of degree $2n$ such that $M_{2n} (E_0) \neq 0$. Hence for $\mu_0=0$, $ \mu^2 = (E-E_0) M_{2n} $ and Corollary \ref{cor-G1}, the equality
 \eqref{eq:chapter3 green funct TD} becomes
\[
\widetilde{g}(E, \mu, x) = \dfrac{i \left ((E-E_0) M_{2n} \big(F_n^0\big)^2 + \frac{(E- E_0)^2 P_n^2}{4} \right ) }{2 \mu (E-E_0) F_n \big(F_n^0\big)^2} = \dfrac{i \left (\frac{M_{2n} }{F_n } + \frac{(E- E_0) P_n^2}{4F_n \big(F_n^0\big)^2} \right ) }{2 \mu }.
\]
Now Corollary \ref{cor-G3} guarantees that
\[
\frac{M_{2n} }{F_n } + \frac{(E- E_0) P_n^2}{4F_n \big(F_n^0\big)^2}
\]
is a polynomial in $E$ of degree $n$. By Lemma \ref{lemma-GH}, we obtain that $\widetilde{n}= n$, $\widetilde{\mu} = \mu$ and $\widetilde{R}_{2\widetilde{n}+1} = R_{2n+1}$. Therefore, for regular points $\widetilde{R}_{2\widetilde{n}+1}$ is a polynomial of degree $2n+1$ in $E$. By Corollary~\ref{cor:curva y pot}, we conclude that $\widetilde{u}$ is solution of a s-KdV$_{n}$ equation. Thus, a Darboux--Crum transformation with a regular point preserves the spectral curve and the level of the s-KdV hierarchy.

Next, we suppose that $(E_0, \mu_0)$ is a singular point of $\Gamma_n$, i.e., $\mu_0 =0$, $R^0_{2n+1} = R_{2n+1} (E_0) $ $=0$ and $R^0_{2n+1, E} = \partial_E ( R_{2n+1} )(E_0) = 0$, thus,
\[ \mu^2 = R_{2n+1} (E) = (E-E_0)^2 Z_{2n-1}, \]
where $Z_{2n-1} (E)$ is a polynomial in $E$ of degree $2n-1$. Hence for $\mu_0=0$, $ \mu^2 = (E-E_0)^2 Z_{2n-1} $ and \ref{cor-G1}, the equality
 \eqref{eq:chapter3 green funct TD} becomes
\[ \widetilde{g}(E, \mu, x) = \dfrac{i \left ((E-E_0)^2 Z_{2n-1} \big(F_n^0\big)^2 + \frac{(E- E_0)^2 P_n^2}{4} \right ) }{2 \mu (E-E_0) F_n \big(F_n^0\big)^2} = \dfrac{i \left (\frac{Z_{2n-1} }{F_n } + \frac{P_n^2}{4F_n \big(F_n^0\big)^2} \right ) }{2 (E-E_0)^{-1} \mu }. \]
Now Corollary \ref{cor-G4} guarantees that
\[
\frac{Z_{2n-1} }{F_n } + \frac{P_n^2}{4F_n \big(F_n^0\big)^2}
\]
is a polynomial in $E$ of degree $n$. By Lemma \ref{lemma-GH}, we obtain that $\widetilde{n} = n-1$ and $\widetilde{\mu} = (E-E_0)^{-1} \mu$. Therefore, $\widetilde{R}_{2\widetilde{n}+1} = (E-E_0)^{-2} R_{2n+1}$ is a polynomial of degree $2n-1$ in $E$. By Corollary \ref{cor:curva y pot}, we conclude that $\widetilde{u}$ is solution of a s-KdV$_{n-1}$ equation. So, a Darboux--Crum transformation with a singular point induces a blow-up in the spectral curve in this singular point and reduces the level of the s-KdV hierarchy in one.
\end{proof}

Next, we will proceed to establish the situation at the point of infinity $P_\infty = [0: 1 :0] $ of the spectral curve. For that,
we will need to work with the Zariski closure in $ \mathbb{P}^2 $ of the spectral curve to understand its behaviour under Darboux transformations for the energy level $ E_0 =0 $. In addition, we will use the blowing-up map in $ \mathbb{P}^2 $ to control the KdV level of the transformed potential $\widetilde{u}$.

Let $\pi\colon \widetilde{\mathbb{P}^2} \to \mathbb{P}^2$ be the blowing-up of $\mathbb{P}^2$ with center $[0:0:1]$. Hence, if $[E:\mu :\nu]$ are homegeneous coordinates in $\mathbb{P}^2$, then the new ones are denoted by $[\widetilde{E}:\widetilde{\mu} :\widetilde{\nu}]$, and $\pi$ is given by
\begin{gather*}
E= \widetilde{E}, \qquad \mu E = \widetilde{\mu}, \qquad \nu = \widetilde{\nu}.
\end{gather*}

\begin{Theorem}[II]\label{thm:curva y DT2} Let $P_\infty = [0: 1 :0] $ be the infinity point of $\overline{\Gamma}_n$, and $u$ a solution of the s-KdV$_n$ equation. Let $\phi_0$ be a solution of Schr\"odinger equation for $P_\infty$ $($in particular $E_0 =0)$ and potential $u$, i.e., $\phi_{0,xx} -u \phi_0 = 0$. Let $\widetilde{u} = u- 2 (\log \phi_0)_{xx}$ be the Darboux--Crum transformation of $u$. Then, $\widetilde{u}$ is solution of the s-KdV$_{n+1}$ equation.
Futhermore, the spectral curve associated to $\widetilde{u}$ is $\Gamma_{n+1} : \widetilde{\mu}^2 - \widetilde{R}_{2n+3}(E) =0$, with $\widetilde{R}_{2n+3} =E^2 R_{2n+1} (E)$.
\end{Theorem}

\begin{proof}
First, consider the homogeneized Green's function associated to the transformed Green's function $\widetilde{g}$. Then, by Propositions \ref{prop-Green2} and \ref{prop-Green3}, $(\widetilde{g})_h $ is a well defined rational function on $\overline{\Gamma}_n$. But also we have
\[ (\widetilde{g})_h = G_h \circ \pi \quad \text{on the spectral curve}. \]
 Moreover $G_h$ is a Green function for the curve defined by $\widetilde{\mu}^2 - \widetilde{R}_{2n+3}(\widetilde{E}) =0$, where $\widetilde{R}_{2n+3} (\widetilde{E}) = E^{2} R_{2n+1} (E)$; that is, for $\Gamma_{n+1} $, the strict transform of $\Gamma_{n}$. Observe that $\widetilde{R}_{2n+3} = E^{2} R_{2n+1}$ is a~polynomial of degree $2n+3$ in $E$. Then, by Corollary \ref{cor:curva y pot}, we conclude that $\widetilde{u}$ is solution of a~s-KdV$_{n+1}$ equation.
\end{proof}

Finally we can rewrite Theorems \ref{thm:curva y DT1} and \ref{thm:curva y DT2} to establish how the spectral curve $\overline{\Gamma}_n$ behaves under Darboux--Crum transformations.

\begin{Theorem}\label{thm:curva y DT}
Let $P=[E_0: \mu_0 :\nu_0]$ be a point in $ \overline{\Gamma}_n $, and $u$ a solution of s-KdV$_n$ equation. Let $\phi_0$ be a solution of Schr\"odinger equation for $E_0$ and potential $u$, say $\phi_{0,xx} = (u- E_0) \phi_0$. Consider $\widetilde{u} = u- 2 (\log \phi_0)_{xx}$ the Darboux--Crum transformation of $u$. Then, $\widetilde{u}$ is solution of s-KdV$_{\widetilde{n}}$ equation for
\[ \widetilde{n} = \begin{cases}
n+1 & \textrm{if } P= [0: 1 :0], \\
n & \textrm{if } P \textrm{ is a regular point of } \Gamma_n,\\
n-1 & \textrm{if } P\textrm{ is an affine singular point of } \Gamma_n.
\end{cases} \]
Futhermore, the spectral curve associated to $\widetilde{u}$ is $\Gamma_{\widetilde{n}}\colon \widetilde{\mu}^2 - \widetilde{R}_{2\widetilde{n} +1} =0$, with
\[ \widetilde{R}_{2\widetilde{n} +1} = \begin{cases}
E^2 R_{2n+1} & \textrm{if } P= [0: 1 :0], \\
R_{2n+1} & \textrm{if } P \textrm{ is a regular point of } \Gamma_n,\\
(E - E_0)^{-2}R_{2n+1} & \textrm{if } P \textrm{ is an affine singular point of } \Gamma_n. \end{cases} \]
\end{Theorem}

\begin{Example} Next we apply the previous theorem to a rational s-KdV$_2$ potential.

Take the s-KdV$_2$ potential $u= \frac{6}{x^2}$ in the Schr\"odinger equation \eqref{eq:schr0 spec}. The spectral curve associated to this potential is $\Gamma_2\colon \mu^2 -E^5 =0$. When $E=0$, we have the fundamental solutions $\phi_1 = x^{-2}$ and $\phi_2 = x^3$. We consider the Darboux transformations of $u$ with these solutions
\begin{gather*}
{\rm DT}(\phi_1) u = u - 2(\log \phi_1)_{xx} = \frac{2}{x^2} = \widetilde{u}_1 \qquad \text{and} \qquad {\rm DT}(\phi_2) u = u - 2(\log \phi_2)_{xx} = \frac{12}{x^2} = \widetilde{u}_3.
\end{gather*}
We have that potential $\widetilde{u}_1$ is a solution of s-KdV$_1$ equation. It is well known that the spectral curve associated to this potential is $\Gamma_1\colon \mu^2 -E^3 =0$, the blowing-up of~$\Gamma_2$ at $(0,0)$. Furthermore, potential $\widetilde{u}_3$ is a solution of s-KdV$_3$ equation, and its associated spectral curve $\Gamma_3$ is the blowing-down of $\Gamma_2$, that is $\Gamma_3\colon \mu^2 -E^7 =0$.

Now, we take a regular value of $E$ in $\Gamma_2$, for instance, $E=-1$. Then, a solution of the Schr\"odinger equation \eqref{eq:schr0 spec} for this value of~$E$ is $\phi^+ = \frac{e^x (x^2 -3x +3)}{x^2}$. The Darboux transformation of $u$ with this solution reads
\[ {\rm DT}(\phi^+) u = u - 2(\log \phi^+)_{xx} = \frac{6(x-1)\big(x^3 -3x^2 +3x -3\big)}{x^2 \big(x^2 -3x +3\big)^2} = \widetilde{u}. \]
Then this transformed potential is a solution of s-KdV$_2$ equation and the spectral curve associated to this potential is still $\Gamma_2\colon \mu^2 -E^5 =0$.

We sum up this example in the following diagram:
\begin{gather*} \xymatrix{ & **[r] (\widetilde{u}_2, \Gamma_3) \ \text{$\begin{array}{l} \text{therefore, $\phi_2$ is} \\ \text{a solution for $P_{\infty}$,}\end{array}$} \\ (u, \Gamma_2) \ar[ur]^{{\rm DT}(\phi_2)} \ar[r]^{\qquad \ {\rm DT}(\phi^+)} \ar[dr]_{{\rm DT}(\phi_1)} & **[r] (\widetilde{u}, \Gamma_2) \ \text{$\begin{array}{l} \text{therefore, $\phi^+$ is a solution} \\ \text{for a regular point,}\end{array}$} \\ & **[r] (\widetilde{u}_1, \Gamma_1) \ \text{$\begin{array}{l} \text{therefore, $\phi_1$ is a solution} \\ \text{for the affine singular point $(0,0)$.}\end{array}$} }
\end{gather*}
\end{Example}

\begin{Remark} The importance of Theorem \ref{thm:curva y DT} lies in the fact that we need to introduce the homogenized Green's function to state it. This new function is the essential tool that allows us to include in our study the point of infinity $P_\infty$ of the affine curve $\Gamma_n$. As far as we know, this is a new approach to the understanding of the spectral curve under Darboux transformations.

Similar problems to our result \ref{thm:curva y DT} were treated by several authors, see \cite[Theorem~5]{EhKn} and~\cite[Theorem~G.2]{GH}. In~\cite{EhKn}, F.~Ehlers and H.~Kn\"orrer studied the action of the Darboux transformations on the spectral curves by means of the eigenfunctions of the centralizer of the Schr\"odinger operator.
\end{Remark}

\subsection[Spectral curves and KdV hierarchy in $1+1$ dimensions]{Spectral curves and KdV hierarchy in $\boldsymbol{1+1}$ dimensions}\label{sec espectral}

In this section we will show how the points of the spectral curves in the stationary setting are related with the solutions of the Schr\"odinger operator with rational potential in the $1+1$ KdV hierarchy.

Recall that the rational soliton $u_{r,n}$ restricted to $t_r = 0$ is the well known $n$-soliton $u_{n}^{(0)} (x) = n(n+1) x^{-2} $. Let $\Gamma_n$ be its affine spectral curve. This complex plane curve has a defining equation
\[
 p_n (E,\mu ) = \mu^2 -E^{2n+1}.
\]
 Our goal was to obtain the algebraic structure of a fundamental matrix of the Schr\"odinger ope\-ra\-tor $-\partial_x^2 +u_{r,n}-E$ by means of the system~\eqref{eq:systemkdv1E rn}. For this purpose we need to use a~parametric representation of the spectral curve $\Gamma_n$. Observe that $\Gamma_n$ is a rational singular plane curve, nevertheless we can have a global parametrization in the sense given in~\cite{BN}. In fact, we have considered the parametrization
\[
 \chi (\lambda ) =\big( {-}\lambda^2, \, i\lambda^{2n+1} \big)
\]
and then $E=-\lambda^2$ as was taken since Section~\ref{sec-fundamental}. Observe that the unique affine singular point of the spectral curve is reached for $\lambda=0$. Hence, whenever $\lambda\not=0$ we obtain regular points on $\Gamma_n$ and we can get the desired description of the fundamental matrix $\mathcal{B}^{(r)}_{n,\lambda}$ as is given in Theorem~\ref{soluciones2 r}. On the other hand, at the singular point $ \chi (0 )=(0,0)$ the fundamental matrix for the system~\eqref{eq:systemkdv1E rn} must be obtained in a specific way, see Theorem~\ref{soluciones1}.

The fundamental solutions $\phi_{1,r,n} (x,t_r)$, $\phi_{2,r,n}(x,t_r)$ obtained in Theorem \ref{soluciones1} were used as source to perform Darboux transformations. In particular, for $t_r=0 $, we get the functions
\[
\phi^{(0)}_{1,n}(x)= \phi_{1,r,n} (x,t_r=0), \qquad \phi^{(0)}_{2,n}(x)=\phi_{2,r,n} (x,t_r=0)
\]
and the corresponding potentials are transformed as is indicated in the following diagram:
\begin{gather} \label{eq:diag DT y curva}
\begin{split}&
\xymatrix @R=1.5em{ u^{(0)}_{n-1} \ar[d] & & \ar[ll]_{{\rm DT}(\phi^{(0)}_{1,n})} u^{(0)}_{n} \ar[rr]^{{\rm DT}(\phi^{(0)}_{2,n})} \ar[d] & & u^{(0)}_{n+1} \ar[d] \\ \Gamma_{n-1} && \Gamma_{n} && \Gamma_{n+1}. }
\end{split}
\end{gather}
This situation is a particular case of a more general one that has been obtained in Theorem~\ref{thm:curva y DT}. The diagram \eqref{eq:diag DT y curva} has its time dependent counterpart (see~\eqref{eq:crum unr} and~\eqref{eq:crum un-1r})
\begin{gather*} 
\xymatrix @R=1.0em{ u_{r,n-1} & & \ar[ll]_{{\rm DT}(\phi_{1,r,n})} u_{r,n} \ar[rr]^{{\rm DT}(\phi_{2,r,n})} & & u_{r,n+1}. }
\end{gather*}
The fundamental matrix $\mathcal{B}^{(r)}_{n,0}$ associated to the functions $\phi_{1,r,n} $ and $\phi_{2,r,n} $ can not be changed by the same Darboux transformations used for the potentials since there is a loss of independent solutions; in fact we have the following diagram:
\begin{gather*} 
\xymatrix @R=0.8em{ & & \phi_{1,r,n} \ar[rr]^{{\rm DT}(\phi_{2,r,n})} & & \phi_{1,r,n+1}, \\ \phi_{2,r,n-1} & & \ar[ll]_{ \ {\rm DT}(\phi_{1,r,n})} \phi_{2,r,n}. & & }
\end{gather*}

On the other hand, whenever the point on the spectral curve is a regular point, that is $\lambda\not=0$, we have obtained the behaviour of the fundamental matrices $\mathcal{B}^{(r)}_{j,\lambda}$, for $j= n-1, n, n+1$, as it is encoded in the following diagram:
 \begin{gather*} 
\xymatrix @R=0.8em{ \phi^+_{r,n-1} & & \ar[ll]_{ \ {\rm DT}(\phi_{1,r,n})} \phi^+_{r,n} \ar[rr]^{{\rm DT}(\phi_{2,r,n}) \ } & & \phi^+_{r,n+1}, \\ \phi^-_{r,n-1} & & \ar[ll]_{ \ {\rm DT}(\phi_{1,r,n})} \phi^-_{r,n} \ar[rr]^{{\rm DT}(\phi_{2,r,n}) \ } & & \phi^-_{r,n+1}.}
\end{gather*}

All these situations are reflected in the time dependent frame coming from the stationary one, as we have seen. In particular, in the lack of specialization process from $\mathcal{B}^{(r)}_{n, \lambda}$ to $\mathcal{B}^{(r)}_{n, 0}$. According to Theorem \ref{thm: det E}, we have that $\det \mathcal{B}^{(r)}_{n, \lambda} = -2 \lambda^{2n+1}$, whereas we have $\det \mathcal{B}^{(r)}_{n, 0} = 2n+1 $.

\begin{Remark}We notice then that, despite functions $\phi^{(0)}_{1,n}$
and $\phi^{(0)}_{2,n}$ are fundamental solutions of the Schr\"odinger equation for $E=0$, they are not solutions for the same point of the spectral curve. Therefore, for each singular point of this spectral curve we can only compute one fundamental solution by means of Darboux transformations.

On the other hand, the stationary functions corresponding to $\phi^+_{r,n}$ and $\phi^-_{r,n}$, namely,
\begin{gather*} (\phi^+_{n})^{(0)} (x, \lambda) = \phi^+_{r,n} (x,\lambda, t_r=0) \qquad \text{and} \qquad (\phi^-_{n})^{(0)} (x, \lambda) = \phi^-_{r,n} (x,\lambda, t_r=0),
\end{gather*} are fundamental solutions at regular points of the spectral curve, since they are solutions of the Schr\"odinger equation for $E\neq 0$. In fact, one of them, say $(\phi^+_{n})^{(0)} (x, \lambda)$, is a solution for the point $(E,\mu)$, and the other one, say $(\phi^-_{n})^{(0)} (x, \lambda)$, is a solution for the conjugated point $(E,-\mu)$ of the spectral curve. Then, for each value of $E=-\lambda^2 $, the fundamental matrix $\mathcal{B}^{(r)}_{n, \lambda}$ shows the solutions at conjugated points on the corresponding spectral curve.
\end{Remark}

Next we have computed an explicit example to illustrate the relationship between spectral curves and KdV hierarchy in $1+1$ dimensions for rational solitons.

 \begin{Example} Consider the case $r=1$
 and $n=2$. Let $u_{1,2} (x,t_1 )=\frac{6x(x^3 -6t_1 )}{(x^3 +3t_1 )^2}$ be the KdV$_1$ rational soliton obtained by taking $(\tau_2,\tau_3 )=(3t_1,0)$. Then, the corresponding stationary potential is given by $u^{(0)}_{2} (x)=u^{(0)}_{1,2} (x) =u_{1,2} (x,t_r = 0)= \frac{6}{x^2}$ (see Lemma~\ref{lema:pot estacionarios}). Its spectral curve is $\Gamma_2\colon p_2 (E,\mu)=\mu^2-E^{5}$.

 Futhermore, the stationary Schr\"odinger operator presents two types of solutions a priori. In fact, when $E=0$, the solutions are
\begin{gather*}
 \phi^{(0)}_{1,2} := \phi_{1,1,2} (x, t_r=0)=x^{-2}, \qquad \phi^{(0)}_{2,2} := \phi_{2,1,2} (x, t_r=0)= x^{3},
\end{gather*}
where
\begin{gather*}
\phi_{1,1,2} (x, t_r)= \dfrac{x}{x^3 + 3t_1}, \qquad \phi_{2,1,2} (x, t_r)=\dfrac{x^6 + 15x^3t_1 -45t_1^2}{x^3 + 3t_1}
\end{gather*}
as they were computed in Example~\ref{ex:E=0}. In this case, we have the following diagram:
\begin{gather*} 
\xymatrix @R=1.5em{ u^{(0)}_{1} =2/x^{2}\ar[d] & & \ar[ll]_{{\rm DT}(\phi^{(0)}_{1,2})} u^{(0)}_{2}=6/x^{2} \ar[rr]^{{\rm DT}(\phi^{(0)}_{2,2})} \ar[d] & & u^{(0)}_{3} =12/x^{2}\ar[d] \\ \mu^2-E^{3}=0 && \mu^2-E^{5}=0 && \mu^2-E^{7}=0. }
\end{gather*}
When energy $E\not=0$, in Example \ref{ex-sol-E} we have computed the solutions
\begin{gather*}
 \phi_{1,2}^+ = e^{\lambda x - \lambda^{3} t_1} \dfrac{\lambda^2 x^3 -3\lambda x^2 +3x +3\lambda^2 t_1}{x^3 +3t_1}, \qquad
 \phi_{1,2}^- = e^{-\lambda x +\lambda^3 t_1} \dfrac{\lambda^2 x^3 +3\lambda x^2 +3x +3\lambda^2 t_1}{x^3 +3t_1},
\end{gather*}
where we have adjusted parameters $\tau_{2}^+ =3 \lambda^2 t_1 =\tau_{2}^-$. Next, take $t_1=0$ to obtain
\begin{gather*}
 \phi_{2}^+ (x,\lambda) =\phi_{1,2}^+ (x,t_r=0,\lambda)= e^{\lambda x } \frac{\lambda^2 x^3 -3\lambda x^2 +3x }{x^3 }, \\
 \phi_{2}^- (x,\lambda) =\phi_{1,2}^- (x,t_r=0,\lambda)= e^{-\lambda x } \frac{\lambda^2 x^3 +3\lambda x^2 +3x }{x^3 }.
\end{gather*}
These functions are solutions of the Schr\"odinger operator for the stationary potential $u^{(0)}_{2}=6/x^{2}$ whenever $E\not=0$. Observe that $\phi_{2}^+ (x,0)= 3/x^2 =\phi_{2}^- (x,0)$, and then they are no longer independent (see Example~\ref{rem:tabla E no 0 taus} for the general case).

Next, we will show how the Darboux transformations act on time dependent potentials and solutions. First recall that for any potential $u$, we have defined the Darboux transformation as
\[
{\rm DT}(\phi_{i,r,n} )u= u-2\left( \log \phi_{i,r,n} \right)_{xx}, \qquad i=1, 2.
\]
Next, we perform the Darboux transformations by means of $\phi_{1,1,2}$ and $\phi_{2,1,2}$ to our initial potential $u_{1,2} $. In these cases we have obtained
\begin{gather*} 
\xymatrix { u_{1,1} =\dfrac{2}{x^2 } & & u_{1,2} =\dfrac{6x\big(x^3\! -6t_1 \big)}{\big(x^3\! +3t_1 \big)^2} \ar[ll]_{{\rm DT}(\phi_{1,1,2}) \quad \ \ } \ar[rr]^{{\rm DT}(\phi_{2,1,2}) \qquad \quad \ } & & u_{1,3} =\dfrac{6x \big(2x^9\! + 675x^3t_1^2\! +1350t_1^3\big)}{\big(x^6\! + 15x^3t_1 -45t_1^2\big)^2}.\!\! }
\end{gather*} 
Then, we must consider the Schr\"odinger operators
\[
-\partial_x^2 +u_{1,j}(x,t_1 )-E, \qquad j=1,2,3.
\]
Their solutions $\phi_{1,j}^+$ and $\phi_{1,j}^-$ were given in Example \ref{ex-sol-E}.

It should be noted that if the energy is not zero, these solutions inherit the same behaviour as their corresponding potentials when the Darboux transformations ${\rm DT}(\phi_{1,1,2})$ and ${\rm DT}(\phi_{2,1,2})$ act on them. Hence we obtain the following diagram:
\begin{gather*}
\xymatrix @R=0.1em{ \phi^+_{1,1}= \dfrac{e^{\lambda x - \lambda^{3} t_1}(\lambda x -1)}{x} & & \ar[ll]_{{\rm DT}(\phi_{1,1,2}) \qquad \quad } \phi^+_{1,2} = \dfrac{e^{\lambda x - \lambda^{3} t_1}(\lambda^2 x^3 -3\lambda x^2 +3x +3\lambda^2 t_1 )}{x^3 +3t_1}\\ \ar[rr]^{\qquad \quad {\rm DT}(\phi_{2,1,2})} & & \phi^+_{1,3} = \dfrac{e^{\lambda x - \lambda^{3} t_1}Q^+_3 (\lambda,x, t_1) }{x^6 + 15x^3t_1 -45t_1^2}, \\
\phi^-_{1,1} = \dfrac{e^{-\lambda x +\lambda^3 t_1}(\lambda x +1)}{x} & & \ar[ll]_{{\rm DT}(\phi_{1,1,2}) \qquad \quad } \phi^-_{1,2} = \dfrac{e^{-\lambda x +\lambda^3 t_1}(\lambda^2 x^3 +3\lambda x^2 +3x +3\lambda^2 t_1)}{x^3 +3t_1} \\
\ar[rr]^{ \quad \qquad {\rm DT}(\phi_{2,1,2})} & & \phi^-_{1,3} = \dfrac{e^{-\lambda x +\lambda^3 t_1} Q^-_3 (\lambda,x, t_1)}{x^6 + 15x^3t_1 -45t_1^2}, }
\end{gather*}
where
\begin{gather*}
Q^+_3 (\lambda,x, t_1) = \lambda^3 x^6 - 6\lambda^2 x^5 +15\lambda x^4 -15x^3 + 15\lambda^3 x^3 t_1 - 45\lambda^2 x^2 t_1 + 45\lambda xt_1 \\
\hphantom{Q^+_3 (\lambda,x, t_1) =}{} - 45\lambda^3 t_1^2 - 45t_1,\\
Q^-_3 (\lambda,x, t_1) = \lambda^3 x^6 + 6\lambda^2 x^5 +15\lambda x^4 +15x^3 + 15\lambda^3 x^3 t_1 +45\lambda^2 x^2 t_1 + 45\lambda xt_1 \\
\hphantom{Q^-_3 (\lambda,x, t_1) =}{} - 45\lambda^3 t_1^2 +45t_1.
\end{gather*}

The zero energy case is essentially different from the point of view of the Darboux transformations. We only can partially obtain the previous diagram:
\begin{gather*} \hspace*{-10mm}\xymatrix @R=0.1em{ & & \phi_{1,1,2} = \dfrac{x}{x^3\! + 3t_1} \ar[rr]^{\!\!\!\!\!\!\!\!\!{\rm DT}(\phi_{2,1,2})} & & \phi_{1,1,3} = \dfrac{x^3\! + 3t_1}{x^6\! + 15x^3t_1\! -45t_1^2}\!\!\! \\ \phi_{2,1,1} = \dfrac{x^3\! +3t_1}{x} & & \ar[ll]_{\!\!\!\!\!\!\!{\rm DT}(\phi_{1,1,2}) \ } \phi_{2,1,2} = \dfrac{x^6\! + 15x^3t_1\! -45t_1^2}{x^3\! + 3t_1}. & & }
\end{gather*}
To compute fundamental matrices associated to $u_{1,1}$ and $u_{1,3}$ we have to use Theorem \ref{soluciones1} (see Example~\ref{ex:E=0}).
\end{Example}

\section{Differential Galois groups }\label{galo}

In this section we study the Picard--Vessiot extensions of the differential systems \eqref{eq:systemkdvr} and \eqref{eq:systemkdv1E r}, obtained for energy levels $E=0$ and $E\not=0$ respectively. We recall that the base differential field is $K_r = \nC ( x, {t}_r )$ with field of constants $\nC$.

We point out that the behaviour that they present depend strongly on the affine point $P=(E,\mu)$ of the corresponding spectral curve.
They present a similar behaviour when the point $P=(E,\mu)$ is a regular point of $\Gamma_n$.

A fundamental matrix for $E = 0$ can be also computed. However, it is not obtained by a~specialization process from the fundamental matrix obtained for a regular point.

We obtain the Picard--Vessiot extensions given by $\mathcal{B}^{(r)}_{n,0}$ and $\mathcal{B}^{(r)}_{n,\lambda}$ and compute their corresponding differential Galois group, say $\mathcal{G}^{(r)}_{n,0}$ and $\mathcal{G}^{(r)}_{n,\lambda}$ respectively.

\subsection[Case $E=0$]{Case $\boldsymbol{E=0}$}

For this case we have the fundamental matrix
\[ \mathcal{B}^{(r)}_{n,0} = \begin{pmatrix} \phi_{1,r,n} & \phi_{2,r,n} \\ \phi_{1,r,n,x} & \phi_{2,r,n,x} \end{pmatrix},\]
where $\phi_{1,r,n}$, $\phi_{1,r,n,x}$, $\phi_{2,r,n}$, $\phi_{2,r,n,x} $ are rational functions in $x$, $t$, hence they are in~$K_r$. So, the Picard--Vessiot field is again~$K_r$. Thus, the differential Galois group is the trivial group, \mbox{$ \mathcal{G}^{(r)}_{n,0} = \{ \textrm{id}_2 \} $}.

\subsection[Case $E \neq 0$]{Case $\boldsymbol{E \neq 0}$}

In this case, we compute the differential extension given for each value of $\lambda \neq 0$. For this, we fix a value of $\lambda$ different from zero, $\lambda = \lambda_0$, then the point $P=(E_0,\mu_0)$ is a regular point of $\Gamma_n$, that is $E_0 \neq 0$. The fundamental matrix is
\[
\mathcal{B}^{(r)}_{n,\lambda_0} = \begin{pmatrix} \phi^+_{r,n} (\lambda_0) & \phi^-_{r,n} (\lambda_0) \\[3pt] \phi^+_{r,n,x} (\lambda_0) & \phi^-_{r,n,x} (\lambda_0) \end{pmatrix},
\]
for $\phi^+_{r,n} (\lambda_0)$, $\phi^+_{r,n,x} (\lambda_0)$, $\phi^-_{r,n} (\lambda_0)$ and $\phi^-_{r,n,x} (\lambda_0) \in K_r(\eta_r)$, with $\eta_r = e^{\lambda_0 x +(-1)^r \lambda_0^{2r+1} t_r}$. Then, the Picard--Vessiot field is $L_r = K_r(\eta_r)$.

To compute the differential Galois group $\mathcal{G}^{(r)}_{n,\lambda_0}$ in this case, we just have to compute the action of $\mathcal{G}^{(r)}_{n,\lambda_0}$ on $\eta_r$. For this, let $\sigma$ in $\mathcal{G}^{(r)}_{n,\lambda_0}$ be an automorphism of the differential Galois group, then
\begin{gather*}
\left( \dfrac{\sigma (\eta_r)}{\eta_r} \right)_x = \dfrac{\sigma (\lambda_0 \eta_r) - \lambda_0 \sigma (\eta_r)}{\eta_r} = \dfrac{\lambda_0 \sigma (\eta_r) - \lambda_0 \sigma (\eta_r)}{\eta_r} = 0,\\
\left( \dfrac{\sigma (\eta_r)}{\eta_r} \right)_{t_r} = \dfrac{\sigma ((-1)^r \lambda_0^{2r+1} \eta_r) - (-1)^r \lambda_0^{2r+1} \sigma (\eta_r)}{\eta_r} \\
\hphantom{\left( \dfrac{\sigma (\eta_r)}{\eta_r} \right)_{t_r}}{} = \dfrac{(-1)^r \lambda_0^{2r+1} \sigma (\eta_r) - (-1)^r \lambda_0^{2r+1} \sigma (\eta_r)}{\eta_r} = 0.
\end{gather*}
Therefore $\frac{\sigma (\eta_r)}{\eta_r}$ is a constant in $K_r$. Hence $\sigma (\eta_r) = c \cdot \eta_r$ for some $c\in \nC$. As a consequence we get that, for each $\lambda_0$ and every $n$, the differential Galois group is isomorphic to the multiplicative group, say
\[
\mathcal{G}^{(r)}_{n,\lambda_0} \simeq G_m = \left \{ \begin{pmatrix} c & 0 \\ 0 & c^{-1} \end{pmatrix}\colon c \in \nC^* \right \}.
\]

\begin{Remark}Since the Galois groups $\mathcal{G}^{(r)}_{n,\lambda_0}$ are obtained for a particular value of $\lambda$ by especialization process, they do not depend on~$\lambda$. For a spectral study of the Picard--Vessiot extensions see \cite{MoRuZu}.
\end{Remark}

\subsection{Global behaviour of the differential Galois groups}

Let us consider the family of linear algebraic groups $\big\{ \mathcal{G}^{(r)}_{n,\lambda} \big\}_{\lambda\in \C}$. Then for each point in $\Gamma_n$ we have found a linear algebraic group. As a result of our constructions we have a sheave structure of groups on the regular points of $\Gamma_n$
\[
\Gamma_n \setminus \text{Sing} (\ \Gamma_n ) \ni \big({-}\lambda^2, i \lambda^{2n+1}\big) \longrightarrow \mathcal{G}^{(r)}_{n,\lambda}.
\]
For each $\lambda \in \nC$, the situation is encoded in the following diagram
\begin{gather*} 
\xymatrix@R=1.4em{ \mathcal{G}_{n-1} \ar[d] & & \mathcal{G}_{n} \ar[d] & & \mathcal{G}_{n+1} \ar[d] \\ \Gamma_{n-1}^* \ar@{<.}[d] \ar[rr]^{\text{Blowing-up} } & & \Gamma_{n}^* \ar[rr]^{\text{Blowing-up}} \ar@{<.}[d] & &\Gamma_{n+1}^* \ar@{<.}[d] \\ L_{n-1} & & \ar[ll]_{{\rm DT}(\phi^{(0)}_{1,n})} L_{n} \ar[rr]^{{\rm DT}(\phi^{(0)}_{2,n})}&& L_{n+1}. }
\end{gather*}

We observe {\it the invariance of the Galois groups with respect to:
\begin{itemize}\itemsep=0pt
\item The time of each level $r$ of the KdV hierarchy once it have been adjusted to the level of the KdV hierarchy. Observe that we are constructing the field of coefficients $K_r$.
\item Generic values of the spectral parameter, i.e., moving along the regular points of the spectral curve.
\item Darboux transformations. \end{itemize}}

\appendix

\section{Auxiliary results}

We establish a series of easy corollaries of the result of Proposition~\ref{prop-Green1}. They are necessary in the Section~\ref{ssec-DT-C}. We use the same notation as in Section~\ref{ssec:ext green}.

\begin{Corollary}\label{cor-G1}
We have
\[
F_n^0 F_{n,x} - F^0_{n,x} F_n = (E - E_0) P_n,
\]
where $P_n$ is a polynomial in $E$ of degree at most $n-1$. In particular for $E_0 =0$ we obtain
\[
f_n F_{n,x} - f_{n,x} F_n = E P_n.
\]
\end{Corollary}
\begin{proof}Since $ F_n = \sum\limits_{l=0}^n f_{n-l} E^l$ and $F^0_n = \sum\limits_{l=0}^n f_{n-l} E_0^l $, we have that
\begin{gather}
 F_n^0 F_{n,x} - F^0_{n,x} F_n = \sum_{i,j=0}^n f_{n-i} f_{n-j,x} E_0^i E^j - \sum_{i,j=0}^n f_{n-i} f_{n-j,x} E_0^j E^i \nonumber \\
\hphantom{F_n^0 F_{n,x} - F^0_{n,x} F_n}{} = \sum_{\substack{i,j=0\\i\ne j}}^n \big(E_0^i E^j - E_0^j E^i\big) f_{n-i} f_{n-j,x}.\label{eq:producto}
\end{gather}
We factor the term $E_0^i E^j - E_0^j E^i$:
\[
E_0^i E^j - E_0^j E^i = (E - E_0) (E E_0)^{\min (i,j)} (-1)^{\textrm{sign} (i,j)} \left (\sum_{k=0}^{|j-i|-1} E^k E_0^{|j-i|-1-k} \right ),
\]
and replace it in \eqref{eq:producto}. We get
\begin{gather*}
F_n^0 F_{n,x} - F^0_{n,x} F_n = \nonumber \\
\qquad{} = (E - E_0) \sum_{\substack{i,j=0\\i\ne j}}^n (E E_0)^{\min (i,j)} (-1)^{\textrm{sign} (i,j)} \left (\sum_{k=0}^{|j-i|-1} E^k E_0^{|j-i|-1-k} \right ) f_{n-i} f_{n-j,x} \nonumber \\
\qquad{} = (E - E_0) P_n,
\end{gather*}
for $P_n$ a polynomial in $E$ of degree at most $n-1$, as it is stated.
\end{proof}

\begin{Corollary}\label{cor-G2}We have
\[ \mu^2 \big(F_n^0\big)^2 -\mu_0^2 F_n^2 = (E - E_0) \left ( \dfrac{F_n F_n^0 P_{n,x}}{2} + F_n^2\big(F_n^0\big)^2 - \dfrac{P_n (F_{n} F_{n,x}^0 + F_{n,x} F_n^0)}{4} \right ),
\]
where $P_n$ is the polynomial obtained in Corollary~{\rm \ref{cor-G1}}. In particular for $E_0 =0$ we obtain
\begin{gather*} \mu^2 f_n^2 -\mu_0^2 F_n^2 = E \left ( \dfrac{F_n f_n P_{n,x}}{2} + F_n^2(f_n)^2 - \dfrac{P_n (F_{n} f_{n,x} + F_{n,x} f_n)}{4} \right ).
\end{gather*}
\end{Corollary}

\begin{proof}By \eqref{eq:spectral curve} we have
\begin{gather*}
\mu^2 = R_{2n+1} = \dfrac{F_n F_{n,xx}}{2} - (u - E) F_n^2 - \dfrac{F_{n,x}^2}{4}, \\
\mu_0^2 = R_{2n+1} (E_0)= \dfrac{F^0_n F^0_{n,xx}}{2} - (u - E_0) \big(F_n^0\big)^2 - \dfrac{\big(F^0_{n,x}\big)^2}{4}.
\end{gather*}
Hence,
\begin{gather*}
 \mu^2 \big(F_n^0\big)^2\! -\!\mu_0^2 F_n^2 = \dfrac{F_n F_n^0}{2} \big(F_{n,xx}F_n^0\! -\! F^0_{n,xx}F_n\big)\! + \dfrac{F_{n}^2 \big(F_{n,x}^0\big)^2\! -\! F_{n,x}^2 \big(F_n^0\big)^2}{4}+ (E \!-\! E_0) F_n^2\big(F_n^0\big)^2 \\
 \hphantom{\mu^2 \big(F_n^0\big)^2\! -\!\mu_0^2 F_n^2}{} = \dfrac{F_n F_n^0}{2} \big(F_{n,xx}F_n^0 - F^0_{n,xx}F_n\big) + (E - E_0) F_n^2\big(F_n^0\big)^2 \\
\hphantom{\mu^2 \big(F_n^0\big)^2\! -\!\mu_0^2 F_n^2=}{} + \dfrac{\big(F_{n} F_{n,x}^0 - F_{n,x} F_n^0\big)\big(F_{n} F_{n,x}^0 + F_{n,x} F_n^0\big)}{4}.
 \end{gather*}
As $ F_n^0 F_{n,xx} - F^0_{n,xx}F_n = (F_n^0 F_{n,x} - F^0_{n,x}F_n)_x = (E-E_0) P_{n,x}$, by Corollary \ref{cor-G1} we obtain
\begin{gather*} \mu^2 \big(F_n^0\big)^2 -\mu_0^2 F_n^2 = (E - E_0) \left ( \dfrac{F_n F_n^0 P_{n,x}}{2} + F_n^2\big(F_n^0\big)^2 - \dfrac{P_n \big(F_{n} F_{n,x}^0 + F_{n,x} F_n^0\big)}{4} \right ).\tag*{\qed}
\end{gather*}\renewcommand{\qed}{}
\end{proof}

Now, let $(E_0, \mu_0)$ be a regular point of $\Gamma_n$ and $\mu_0 = 0$. In this case, we have that $R^0_{2n+1} = R_{2n+1} (E_0) =0$ and $ \partial_E ( R_{2n+1} )(E_0) \neq 0$, thus,
\begin{gather}\label{eq:M2n}
 \mu^2 = R_{2n+1} (E) = (E-E_0) M_{2n},
\end{gather}
where $M_{2n} (E)$ is a polynomial in $E$ of degree $2n$ such that $M_{2n} (E_0) \neq 0$.

\begin{Corollary}\label{cor-G3}
Let $(E_0, \mu_0)$ be a regular point of $\Gamma_n$ and $\mu_0 = 0$. We have that
\[
\frac{M_{2n} }{F_n } + \frac{(E- E_0) P_n^2}{4F_n \big(F_n^0\big)^2}
\]
is a polynomial in $E$ of degree $n$, with $P_n$ the polynomial obtained in Corollary~{\rm \ref{cor-G1}} and $M_{2n}$ the polynomial defined in \eqref{eq:M2n}.
\end{Corollary}

\begin{proof}We have
\begin{gather*}
 M_{2n} = \dfrac{\mu^2}{E-E_0} = \dfrac{F_n F_{n,xx}}{2(E-E_0)} - \dfrac{(u-E) F^2_n}{E-E_0} - \dfrac{F^2_{n,x}}{4(E-E_0)}, \\
 P^2_n = \dfrac{\big(F^0_{n} F_{n,x} - F^0_{n,x} F_{n}\big)^2}{(E-E_0)^2} = \dfrac{\big(F^0_{n}\big)^2 F^2_{n,x} + \big(F^0_{n,x}\big)^2 F^2_{n} -2 F^0_{n}F_{n} F_{n,x} F^0_{n,x} }{(E-E_0)^2}.
\end{gather*}
We replace these expressions in the formula and we get
\[ \frac{M_{2n} }{F_n } + \frac{(E- E_0) P_n^2}{4F_n \big(F_n^0\big)^2} = \dfrac{2 \big(F^0_n\big)^2 F_{n,xx} -4(u-E) \big(F^0_n\big)^2 F_n + \big(F^0_{n,x}\big)^2 F_n - 2 F^0_n F^0_{n,x} F_{n,x}}{4(E-E_0) \big(F^0_n\big)^2}. \]
The numerator of this function is a polynomial in $E$ of degree $n+1$ and has a root in $E=E_0$ as can be easily verified replacing $E$ by $E_0$:
\[ 2 \big(F^0_n\big)^2 F^0_{n,xx} -4(u-E^0) \big(F^0_n\big)^3 - \big(F^0_{n,x}\big)^2 F^0_n = 4 F^0_n \mu_0^2 =0. \]
So, we get that
\[ 2 \big(F^0_n\big)^2 F_{n,xx} -4(u-E) \big(F^0_n\big)^2 F_n + \big(F^0_{n,x}\big)^2 F_n - 2 F^0_n F^0_{n,x} F_{n,x} = (E-E_0) Q_n, \]
where $Q_n$ denotes a polynomial in $E$ of degree $n$. Hence
\[
\frac{M_{2n} }{F_n } + \frac{(E- E_0) P_n^2}{4F_n \big(F_n^0\big)^2} = \dfrac{ Q_n}{4\big(F^0_n\big)^2}
\]
and then the result follows.
\end{proof}

Next, let $(E_0, \mu_0)$ be a singular point of $\Gamma_n$. In this case, $\mu_0 =0$, $R^0_{2n+1} = R_{2n+1} (E_0) $ $=0$ and $ \partial_E ( R_{2n+1} )(E_0) = 0$, thus,
\begin{gather}\label{eq:Z2n}
 \mu^2 = R_{2n+1} (E) = (E-E_0)^2 Z_{2n-1},
\end{gather}
where $Z_{2n-1} (E)$ is a polynomial in $E$ of degree $2n-1$ such that $Z_{2n-1} (E_0) \neq 0$.

\begin{Corollary}\label{cor-G4}
Let $(E_0, \mu_0)$ be a singular point of $\Gamma_n$. We have that
\[
\frac{Z_{2n-1} }{F_n } + \frac{P_n^2}{4F_n \big(F_n^0\big)^2}
\]
is a polynomial in $E$ of degree $n-1$, with $P_n$ the polynomial obtained in Corollary~{\rm \ref{cor-G1}} and $Z_{2n-1}$ the polynomial defined in~\eqref{eq:Z2n}.
\end{Corollary}

\begin{proof}
It follows by an analogous computation to that of Corollary \ref{cor-G3}.
\end{proof}

\subsection*{Acknowledgements}

We kindly thank all members of the Integrability Madrid Seminar for many fruitful discussions: P.~Acosta-Hum\'anez, D.~Bl\'azquez, J.A.~Capit\'an, R.~Hern\'andez Heredero, A.~P\' erez-Raposo, J.~Rojo Montijano and S.~Rueda. Authors gratefully acknowledge the reviewers for their helpful comments and further references which resulted in an improvement of the preliminary manuscript.

\pdfbookmark[1]{References}{ref}
\LastPageEnding

\end{document}